\documentclass[journal]{IEEEtran}

\usepackage[left=0.625in,right=0.625in,top=0.75in,bottom=1in]{geometry}
\usepackage{amsmath}
\usepackage{amssymb}
\usepackage{graphicx}
\usepackage{array}
\usepackage{algorithm}
\usepackage{algpseudocode}
\usepackage{xcolor}
\usepackage{cite}
\usepackage{url}
\usepackage{hyperref}
\usepackage[utf8]{inputenc}
\usepackage{amsthm}
\usepackage{stfloats}
\usepackage{comment}
\usepackage{eso-pic}
\hypersetup{}

\newtheorem{proposition}{Proposition}[section]
\newtheorem{corollary}{Corollary}[section]

\begin{document}

\title{\fontsize{18pt}{18pt}\selectfont Defense-in-Depth at the Perception-Reasoning Interface of LLM-Centric Agentic UAV Swarms}

\author{ Mohammadhossein~Homaei,~\IEEEmembership{ Member,~IEEE,}
Yousef~Emami,~\IEEEmembership{Senior Member,~IEEE,}
Sajad~Homayoun,~\IEEEmembership{Senior Member,~IEEE,}
Rahim~Taheri,~\IEEEmembership{Senior Member,~IEEE,}
Hao~Zhou,~\IEEEmembership{Senior Member,~IEEE,}
Miguel~Gutiérrez~Gaitán,~\IEEEmembership{Senior Member,~IEEE,}
Bo~Wei,~\IEEEmembership{Senior Member,~IEEE,}

% \thanks{Mohammadhossein Homaei is with the Departamento de Ingeniería de Sistemas Informáticos y Telemáticos, Universidad de Extremadura, Av/ Universidad S/N, Cáceres, 10003, Spain. (e-mail: homaei@ieee.org and homaei@unex.es)}

% \thanks{Yousef Emami is with Real-Time and Embedded Computing Systems Research Centre (CISTER), 4200-135 Porto, Portugal (email:emami@isep.ipp.pt)}

% \thanks{Sajad Homayoun is with the Department of Electronic Systems, Aalborg University, Denmark (email:sajadh@es.aau.dk)}

% \thanks{Rahim Taheri is with the School of Computing, Maths and Physics, University of Portsmouth, UK (e-mail: rahim.taheri@port.ac.uk)}

% \thanks{Hao Zhou is with the School of Computer Science, McGill University, Montreal, QC H3A 0E9, Canada. (email:hzhou098@uottawa.ca)}

% \thanks{Miguel Gutiérrez Gaitán is with Pontificia Universidad Católica de Chile (email:miguel.gutierrez@uc.cl)}

% \thanks{Bo Wei is with the School of Computing, Newcastle University, UK (email:bo.wei@newcastle.ac.uk)
% }
}
\AddToShipoutPictureBG*{%
  \AtPageLowerLeft{%
    \put(0,\LenToUnit{0.4in}){%
      \makebox[\paperwidth]{%
        \parbox{0.85\paperwidth}{\centering\footnotesize
          This work has been submitted to the IEEE for possible publication.
          Copyright may be transferred without notice, after which this
          version may no longer be accessible.}}}}}
\maketitle

\begin{abstract}
Large Language Models (LLMs) increasingly support Uncrewed Aerial Vehicle (UAV) swarm operations such as data collection scheduling, where the model reads structured sensor reports and decides which sensors to visit. An adversary who quietly manipulates those reports can redirect the swarm without modifying the model weights or the UAV. Defenses for this interface have been proposed architecturally but rarely implemented or evaluated. We implement and evaluate defense-in-depth at the perception-reasoning interface of LLM-Centric Agentic UAV Swarms. Five layers check the provenance of a report, whether its values are physically admissible, whether they agree with what swarm geometry and service history predict, whether the resulting schedule starves any sensor, and, when these fail, hand control to a deterministic scheduler that ignores the suspect input. We test each layer against an adversary strong enough to defeat the layer before it. For each of the three input-side layers, we derive in closed form how far a report can be distorted before that layer reacts, fixing each boundary from deployment parameters before any attack data is collected; across thirty matched simulation runs, predicted and measured boundaries agree. 
Separating attack detection from response is a well-established principle, and we quantify the cost of neglecting this distinction at the perception–reasoning interface. When the system rejects a report, it replaces it with the most recent accepted report. This prevents the adversary from controlling the UAV schedule, but it also increases cumulative cost by 79\% and 74\% for the two detectors, respectively, compared with the undefended system. The safety check does not detect any attacks, but it nevertheless reduces the attack-induced cost by 37.5\%.
\end{abstract}

\IEEEpeerreviewmaketitle

\section{Introduction}

Uncrewed Aerial Vehicle (UAV) swarms coordinate sensing, communication, and decision-making for surveillance, search and rescue, environmental monitoring, disaster response, and infrastructure inspection. Integrating Artificial Intelligence (AI) into this coordination couples perception directly to decision-making, so an attacker who manipulates the information available to the controller can influence subsequent actions, and dependable operation under adversarial conditions becomes a deployment requirement \cite{alqudsi2025uav}. Agentic AI lets such systems pursue complex goals with limited supervision \cite{Sapkota2025}, but this autonomy also broadens the attack surface \cite{Dehghantanha2026} to prompts, retrieved knowledge, agent memory, tool interactions, and perception inputs, so security must cover the agentic pipeline rather than the Large Language Model (LLM) in isolation.

The concern is acute when LLMs schedule data collection, routing, or task allocation inside the control loop. The model does not observe the network directly: it reads a structured report of each sensor's backlog, residual energy, and link quality without verifying these values, which makes the perception–reasoning interface a critical security boundary. The Priority Manipulation Attack (PMA) reported for the LLM-Centric Agentic AI for UAV Swarm (LAUS) architecture \cite{emami2026laus} shows that an adversary can manipulate otherwise plausible observation data, without modifying model weights or the physical system, to redirect swarm behavior and substantially increase packet loss. Because such observations remain syntactically valid and physically plausible, schema-based validation may miss the attack, and similar interface-level weaknesses in LLM-assisted robotics, tool protocols, multi-agent collaboration, and agent memory show that defenses at this interface must be evaluated empirically rather than assumed from architectural separation.

\par
We protect this interface with a five-layer Defense-in-Depth (DID) architecture that is implemented and evaluated. $L_1$ verifies report provenance, freshness, and physical admissibility. $L_2$ compares reported path loss and energy with independent predictions from UAV trajectories, sensor positions, and service history. $L_3$ applies queue conservation to detect backlog changes inconsistent with admissible arrivals. At the decision level, $L_4$ repairs schedules that concentrate service or starve sensors, while $L_5$ maintains a tamper-evident audit log and switches to a deterministic scheduler when rejections persist. Removing one layer at a time against a fixed attack would leave later layers untested because $L_1$ blocks a transit-only adversary. Instead, each layer is tested against an adversary capable of bypassing all preceding layers (Fig.~\ref{fig:ladder}). This includes perception compromise that passes $L_1$, backlog-only manipulation that bypasses $L_2$, and an adversary that remains within all admissible regions, where the defense can only limit damage. For each report-side layer, we derive a closed-form manipulation bound and evaluate detection and damage separately, since preventing schedule control does not necessarily prevent performance degradation.

\begin{figure}[!tb]
\centering
\includegraphics[width=0.85\columnwidth]{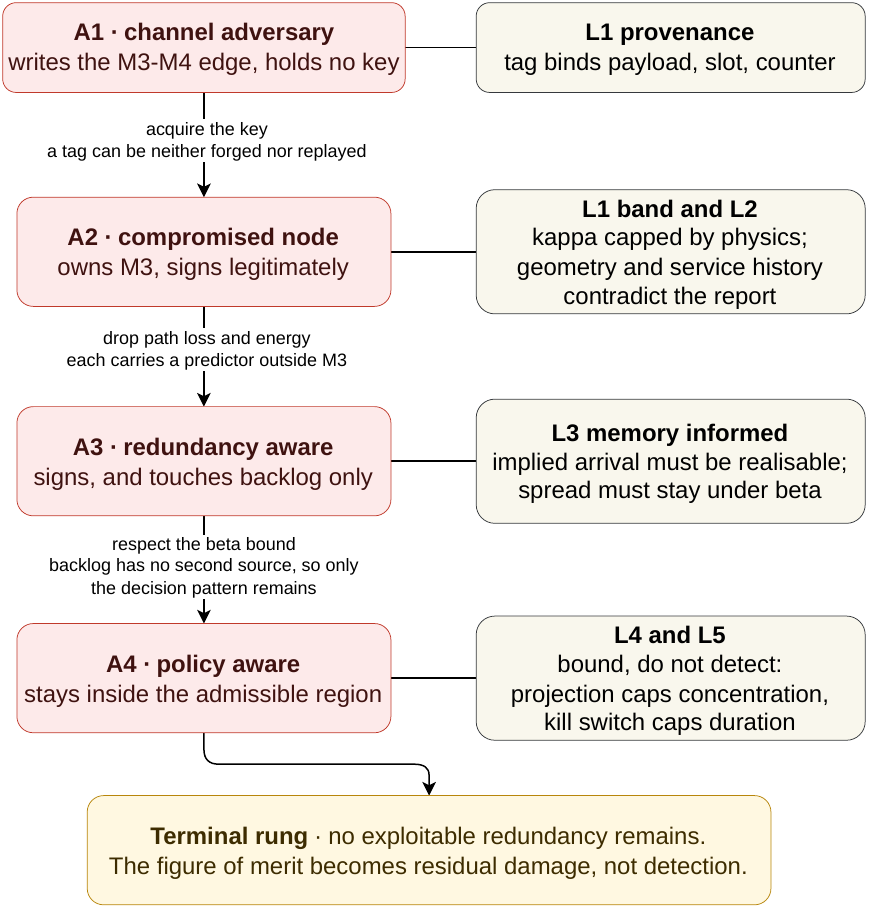}
\caption{Adversary capability ladder and corresponding defense layers. Arrows indicate the capability needed to bypass each layer; at $\mathcal{A}_4$, no exploitable redundancy remains and evaluation shifts from detection to residual damage.
}
\label{fig:ladder}
\end{figure}

\begin{itemize}
\item A DID architecture stated as an operational rather than structural specification, in which each of the five layers is a computable predicate over quantities the pipeline already holds, given with its parameters, its complexity, the adversary class it terminates, and the class it provably does not.

\item An adversary ladder for the perception-reasoning interface whose four classes follow from the redundancy structure of the pipeline: each class removes the redundancy exploited by the previous layer, and the ladder terminates once no exploitable redundancy remains.

\item Three closed-form bounds on adversary budget, one per input-side layer, computable from deployment parameters alone: the channel model caps attainable intensity at $\bar\kappa=1.324$, and at $1.019$ on the realized geometry; the sensing-noise floor places the cross-modal boundary at $\kappa^{\mathrm{det}}_\ell=1.026$; and the queue recursion caps a backlog-only adversary at $\bar\kappa_q=1.2$ once its target is served. Together they place the nominal PMA intensity outside every admissible region and leave a window whose width is quantified rather than asserted.

\item A calibration protocol and a paired evaluation on a locally hosted open-weight model, reporting cost, attack success, per-layer detection and false-positive rates, detection curves, the clean-condition cost of the stack and latency, with thresholds fixed on seeds disjoint from evaluation and the two uncalibrated bounds specified as operator budgets.
\end{itemize}

\section{Related Work} \label{sec:related}

%Recent research has increasingly explored the integration of LLMs, in-context learning (ICL), and agentic AI to enhance the autonomy, intelligence, and security of UAV systems. These developments have expanded the role of LLMs from standalone reasoning models to components of autonomous systems that support perception, planning, coordination, scheduling, and security. At the same time, the increasing integration of LLMs into UAV control loops introduces new security challenges, particularly when decisions depend on information supplied by external perception and communication components.
\subsection{LLM-Assisted UAV Autonomy and Swarm Coordination}

LLMs have been increasingly explored for UAV scheduling, resource allocation, and swarm coordination. Emami et al.~\cite{emami2025icldc} introduce In-Context Learning-based Data Collection (ICLDC), which uses an LLM to generate UAV data-collection schedules from natural-language task descriptions and a rule-based verifier to prevent unsafe actions. The system achieves lower cumulative packet loss than Deep Q-Network (DQN) and maximum-channel-gain baselines, but its vulnerability to manipulated task descriptions highlights the security risks of LLM-assisted decision-making. Flight Resource Allocation Scheme based on LLM-enabled ICL (FRSICL)~\cite{11399950} extends in-context learning to real-time flight resource allocation by jointly optimizing UAV data-collection schedules and velocity using environmental feedback, achieving effective Age of Information (AoI) performance without extensive Deep Reinforcement Learning (DRL) retraining. For multi-UAV coordination, Role-Adaptive LLM-driven Yoked Navigation (RALLY)~\cite{11168174} combines LLM-assisted semantic reasoning, structured communication, dynamic role switching, and Multi-agent Reinforcement Learning (MARL) to improve collaborative navigation and role assignment. At the system level, LAUS~\cite{emami2026laus} integrates perception, memory, reasoning, planning, and action into a closed-loop LLM-centric architecture for autonomous UAV swarms. However, it also highlights the expanded attack surface of such systems, including PMAs that can alter observations and influence swarm decisions and network performance. Together, these studies demonstrate the potential of LLMs and agentic architectures for adaptive UAV autonomy while highlighting the need to secure the information exchanged between perception and reasoning.

\subsection{Security and Robustness of LLM-Enabled UAV Systems}

Security and robustness have become important concerns as LLMs are integrated into autonomous UAV systems. Ferrag et al.~\cite{11397656} survey LLM reasoning and agentic AI for autonomous driving and UAV applications, identifying threats including prompt injection, data poisoning, multimodal perturbations, privacy leakage, and vulnerabilities in cooperative agent communication. For aerial Internet of Things (IoT) security, Li et al.~\cite{li2026multi} develop an LLM-enabled multi-agent intrusion detection framework that uses specialized agents for data processing and adaptive attack classification, achieving over 90\% classification accuracy across benchmark datasets; however, it focuses on network-level intrusion detection rather than manipulation of information supplied to an LLM controller. Ayana et al.~\cite{ayana2025privllmswarm} introduce PrivLLMSwarm, which applies Secure Multi-Party Computation to protect LLM inference for UAV swarm coordination, demonstrating a privacy--utility tradeoff while addressing confidentiality rather than adversarial observation manipulation. For navigation security, Tariq et al.~\cite{tariq2026mission} suggest OrbitGuardNet, combining transformers, Graph Neural Networks (GNNs), Kalman-filter residuals, and an LLM to detect and respond to BeiDou spoofing; although it reports high detection performance in controlled simulations and three-UAV flight tests, its evaluation relies on software-injected rather than open-air Radio Frequency (RF) spoofing. Finally, Aero-LLM~\cite{dharmalingam2025aero} distributes specialized LLMs across onboard, edge, and cloud platforms for inference, anomaly detection, and forecasting, demonstrating the potential of distributed intelligence for UAV security and mission operations. Overall, these studies address important aspects of LLM security, privacy, intrusion detection, and navigation resilience, but do not specifically protect the structured observations exchanged between UAV perception and LLM reasoning.

\subsection{LLM-Assisted UAV Scheduling and Security}

Recent studies have applied LLMs to UAV scheduling and autonomous decision-making, while also identifying vulnerabilities at the perception--reasoning interface \cite{emami2025icldc,emami2026laus}. However, existing evaluations consider different attack surfaces, including prompt-level attacks and manipulation of agent observations, rather than systematically protecting structured sensor states used for resource scheduling. At lower layers of UAV systems, security has been studied through Global Navigation Satellite System (GNSS) spoofing \cite{she2026gnssresilience,geng2025gnsscountermeasure}, jamming \cite{xiang2024jamming}, control-loop attacks \cite{qian2026learningsecurity,zhao2026formation}, and false-data injection \cite{mughal2026gnnids}. Provenance and trust mechanisms \cite{jaigirdar2023provenance,nasiri2024trustdse,akram2024digitaltwintrust} establish trust in devices or data sources, while output-side guardrails \cite{ravichandran2026roboguard,xiang2025guardagent,shamsujjoha2025swisscheese,safiullin2027guardraileval} constrain unsafe agent actions. Memory-poisoning studies \cite{yazdinejad2026memorypoisoning,jing2026dsrm,zhang2025asb} further show that persistent agent state can be exploited.

\subsection{Research Gap}

The above studies demonstrate significant progress in applying LLMs and agentic AI to UAV scheduling, swarm coordination, navigation, intrusion detection, privacy preservation, and cybersecurity; however, limited work specifically protects the perception--reasoning interface, where structured information describing sensor states, communication links, energy, queues, and other operational variables is provided to the LLM. The PMA and ICLDC show that manipulated observations and task inputs can redirect LLM-assisted UAV decisions. Other studies address intrusion detection, privacy preservation, navigation spoofing, and general LLM robustness \cite{11397656,li2026multi,ayana2025privllmswarm,tariq2026mission,dharmalingam2025aero}, but they do not establish a DID mechanism that protects structured state information before it reaches the reasoning process. Moreover, existing approaches do not systematically evaluate how defenses behave as adversarial capabilities increase, such as when an attacker can compromise the perception module, reproduce valid provenance information, or manipulate quantities lacking an independent prediction. This work addresses this gap through an implemented five-layer DID architecture that combines report provenance and physical-admissibility checks, cross-modal consistency analysis, queue-conservation analysis, schedule-level safety checking, and deterministic fallback control, with each layer implemented as a computable predicate using information available within the UAV pipeline.

\section{System Model and Threat Formulation}\label{sec:problem}

\subsection{Network and Mission Model}\label{subsec:network}

A swarm of $U$ UAVs, indexed by $\mathcal{U}=\{1,\dots,U\}$, serves $N$ stationary ground sensors, indexed by $\mathcal{N}=\{1,\dots,N\}$, over a horizon $\mathcal{T}=\{1,\dots,T\}$ of slots of duration $\Delta t$. Sensor $n$ occupies a fixed position $c_n\in\mathbb{R}^2$ inside a square region of side $D$, and UAV $u$ occupies $p_u(t)\in\mathbb{R}^2$ at constant altitude $H$, subject to $\|p_u(t{+}1)-p_u(t)\|\le v_{\max}\Delta t$.

The UAVs fly predetermined trajectories whose waypoints $b_u(t)$ cover the sensors within a service radius $r_{\mathrm{exc}}$. Such trajectories may be designed for capacity, coverage, or propulsion energy, and the approach applies to any given trajectory.% We denote by $b_u(t)$ the waypoint of UAV $u$ at slot $t$, with the reference trajectory advancing by $a$ per slot and fixed at deployment independently of the schedule; within a slot the UAV leaves $b_u(t)$ for the service point $p_u(t)$ of its assigned sensor and returns, so the round trip limits the service radius to $r_{\mathrm{exc}}=(v_{\max}\Delta t-a)/2$.

Writing $d_{n}(z)=\sqrt{H^{2}+\|z-c_n\|^{2}}$ for the separation and $\theta_{n}(z)=\frac{180}{\pi}\arctan(H/\|z-c_n\|)$ for the elevation angle, air-to-ground links follow the probabilistic Line-of-Sight (LoS) model of~\cite{alhourani2014optimal}, in which the elevation angle determines the LoS probability
\begin{equation}\label{eq:plos}
P^{\mathrm{LoS}}_{n}(z)=\frac{1}{1+a_{0}\exp\!\big(\!-b_{0}\,(\theta_{n}(z)-a_{0})\big)},
\end{equation}
with environment coefficients $(a_0,b_0)$, and the LoS and non-LoS excess losses combine into the average path loss
\begin{equation}\label{eq:pathloss}
\begin{split}
L_{n}(z)={}&20\log_{10}\!\Big(\tfrac{4\pi d_{n}(z)}{\lambda}\Big)
 + P^{\mathrm{LoS}}_{n}(z)\,\eta_{\mathrm{LoS}}\\
&+\big(1-P^{\mathrm{LoS}}_{n}(z)\big)\,\eta_{\mathrm{NLoS}},
\end{split}
\end{equation}
where $\lambda$ is the carrier wavelength. The distinction between two evaluation points is load bearing: the reported path loss is evaluated at $z=b_u(t)$, before any assignment exists, and the link-success mapping at $z=p_u(t)$, after it. %Both reduce to the lane of~\eqref{eq:baseline} and the coordinate $c_n$, fixed at deployment and neither produced by $M_3$ nor writable by any adversary here, which lets the geometric predictor of~\eqref{eq:lgeo} be evaluated independently of a compromised perception node.
A transmission attempt on link $(u,n)$ succeeds with probability
\begin{equation}\label{eq:success}
\rho_{u,n}(t)=\Phi\big(L_{n}(p_u(t));\,\tau,k_1,k_2\big)\in[0,1],
\end{equation}
in which $\Phi$ is the monotonically decreasing link-success mapping of the LAUS simulator, parameterized by the loss threshold $\tau$ and link constants $(k_1,k_2)$ and calibrated to a target bit error rate $\mathrm{BER}^{\star}$.

The state of sensor $n$ is the triple $s_n(t)=(q_n(t),e_n(t),\ell_n(t))$, with backlog $q_n(t)\in[0,Q_{\max}]$, residual energy $e_n(t)\in[0,E_0]$, and path loss $\ell_n(t)=\min_{u}L_{n}(b_u(t))$. The scheduler never reads $s_n(t)$ but the report of~\eqref{eq:obs}, which equals $s_n(t)$ up to sensing noise in the absence of an adversary. A scheduling action assigns each UAV to exactly one sensor,
\begin{equation}\label{eq:action}
x(t)=\big[x_{u,n}(t)\big]\in\{0,1\}^{U\times N},
\quad
\sum_{n\in\mathcal{N}}x_{u,n}(t)=1\;\;\forall u,
\end{equation}
and $\mathcal{X}$ is the set of matrices satisfying this constraint. A visited sensor is served at rate $R$ packets per visit, so the number of packets drained is
\begin{equation}\label{eq:service}
\begin{split}
\mu_n(t)=R\,&\mathbb{1}\Big[\textstyle\sum_{u}x_{u,n}(t)\ge 1\Big]\\
\times\,&\mathbb{1}\big[\zeta_{n}(t)\le\max\nolimits_{u:\,x_{u,n}(t)=1}\rho_{u,n}(t)\big],
\end{split}
\end{equation}
where $\mathbb{1}[\cdot]$ is the indicator and $\zeta_n(t)\sim\mathcal{U}[0,1]$ is drawn from the environment stream. Queue and energy evolve as
\begin{align}
q_n(t{+}1)&=\min\Big\{\big[q_n(t)-\mu_n(t)\big]^{+}+\alpha_n(t),\;Q_{\max}\Big\},\label{eq:queue}\\
e_n(t{+}1)&=\big[e_n(t)-\epsilon_{\mathrm{tx}}\mu_n(t)-\epsilon_{\mathrm{idle}}\big]^{+},\label{eq:energy}
\end{align}
with exogenous arrivals $\alpha_n(t)$. The distinction between the endogenous service term and the exogenous arrival term later determines where the adversary ladder terminates.

The realized service $\mu_n(t)$, used by two layers, is not a sensor report: it is formed at the serving UAV from its own delivery receipts and reaches the defense over the swarm's control link rather than the $M_3\!\rightarrow\!M_4$ edge. This matters because~\eqref{eq:emem} and~\eqref{eq:impliedarrival} are affine in $\mu_n$: a class owning $M_3$ that could also report $\tilde\mu_n=\kappa\mu_n-(\kappa-1)\alpha_n$ alongside a backlog scaled by $\kappa$ would make the implied arrival equal the true $\alpha_n$ exactly, rendering Proposition~\ref{prop:kappaq} and Corollary~\ref{cor:noadmissible} vacuous.

Two events incur cost, a packet dropped through overflow and a failed transmission attempt, so
\begin{align}
C(t)&=\sum_{n\in\mathcal{N}}
\Big(\big[\,\overline{q}_n(t)-Q_{\max}\big]^{+}+\nu_n(t)\Big),\label{eq:cost}\\
J&=\sum_{t\in\mathcal{T}}C(t),\label{eq:objective}
\end{align}
where $\overline{q}_n(t)=[q_n(t)-\mu_n(t)]^{+}+\alpha_n(t)$ is the pre-truncation backlog implied by~\eqref{eq:queue} and $\nu_n(t)$ counts failed attempts, each charged as one lost packet. The objective is $\min_{x(1:T)}J$, and $J$ is the primary metric consistent with~\cite{emami2026laus}.

\subsection{Agentic Decision Pipeline}\label{subsec:pipeline}

The scheduler is the closed-loop agent of~\cite{emami2026laus}, whose decision pipeline decomposes into six modules (Fig.~\ref{fig:pipeline}): $M_1$ and $M_2$ acquire telemetry and extract features, $M_3$ serializes them into a structured observation, $M_4$ assembles the prompt, $M_5$ reasons and plans, and $M_6$ validates and actuates. The observation produced by $M_3$ is
\begin{equation}\label{eq:obs}
o(t)=\Big\{\big(\hat q_n(t),\;\hat e_n(t),\;\hat\ell_n(t)\big)\Big\}_{n\in\mathcal{N}},
\end{equation}
where a hat denotes a reported quantity, and $\sigma(t)=\mathrm{ser}(o(t))$ is the serialized form carried on the $M_3\!\rightarrow\!M_4$ edge.

Reported quantities are measurements, not ground truth, and in the absence of an adversary satisfy
\begin{equation}\label{eq:noise}
\hat\phi_n(t)=\phi_n(t)+\nu^{\phi}_n(t),
\qquad \nu^{\phi}_n(t)\sim\mathcal{N}\big(0,(\sigma_\phi)^2\big),
\end{equation}
where $\sigma_\ell$ is the receiver's path-loss estimation error after averaging over a packet burst and $\sigma_e$ the readout error of the battery gauge. We take $\sigma_\ell=0.5$~dB, a sub-decibel receiver accuracy, and $\sigma_e=0.05$~J, one idle-slot drain. UAV deviation from its waypoint affects the predictor of~\eqref{eq:lgeo} rather than the measurement and is treated as model mismatch in~\eqref{eq:kappamis}. This noise sets the scale of the $L_2$ residuals, which Proposition~\ref{prop:l2detect} converts into a detection boundary that a noisier receiver moves rather than invalidates, from $\kappa^{\mathrm{det}}_\ell=1.024$ at $0.5$~dB to $1.050$ at $1$~dB.

Agentic memory retains a rolling window of $W_m$ past decision records,
\begin{equation}\label{eq:memory}
\mathcal{M}(t)=\Big\{\big(\sigma(t{-}j),\;x(t{-}j),\;C(t{-}j)\big)\Big\}_{j=1}^{W_m},
\end{equation}
where $W_m$ sets both the history the prompt carries and the horizon of the cumulative energy residual of~\eqref{eq:ecum}. An energy drift of $\epsilon$ per slot leaves a one-step residual of exactly $\epsilon$ but accumulates to $W_m\epsilon$ across the window, whereas under~\eqref{eq:noise} the clean cumulative residual is the difference of two independent gauge readings, with threshold $z_N\sqrt{2}\,\sigma_e$ irrespective of $W_m$. The cumulative branch therefore dominates the one-step branch precisely when
\begin{equation}\label{eq:wmcond}
W_m>\frac{z_N\sqrt{2}\,\sigma_e}{\epsilon},
\end{equation}
which at a drift of one gauge quantum gives $W_m>4.92$, while a $4096$-token context with a $2048$-token reasoning budget leaves room for the directive, the $60$ scalars of an $N=20$ observation, and about five records; hence $W_m=5$. Module $M_4$ assembles the prompt $\pi(t)=\mathrm{asm}(\sigma(t),\mathcal{M}(t),g)$ from the serialized observation, the memory and the static mission directive $g$, $M_5$ maps it to a candidate assignment, and $M_6$ applies the output-side checks of Section~\ref{sec:did},
\begin{equation}\label{eq:policy}
\tilde x(t)=\mathcal{G}_{\theta}\big(\pi(t)\big),
\qquad
x(t)=M_6\big(\tilde x(t),\;\mathcal{M}(t)\big),
\end{equation}
which closes the loop through~\eqref{eq:service} and~\eqref{eq:memory}. A paired comparison measures the defense rather than decoding noise only if the realized policy $\mathcal{G}^{\dagger}_{\theta}$ is single-valued,
\begin{equation}\label{eq:determinism}
\pi(t)=\pi'(t)\;\Longrightarrow\; \mathcal{G}^{\dagger}_{\theta}\big(\pi(t)\big)
=\mathcal{G}^{\dagger}_{\theta}\big(\pi'(t)\big).
\end{equation}
Greedy decoding is necessary but not sufficient on a batched engine, where different batch neighbors may change the floating-point reduction order, so the property is enforced by memoizing $\pi\mapsto\tilde x$, leaving the seed-indexed environment stream as the only stochasticity in a run.

\begin{figure}[!tb]
\centering
\includegraphics[width=0.78\columnwidth]{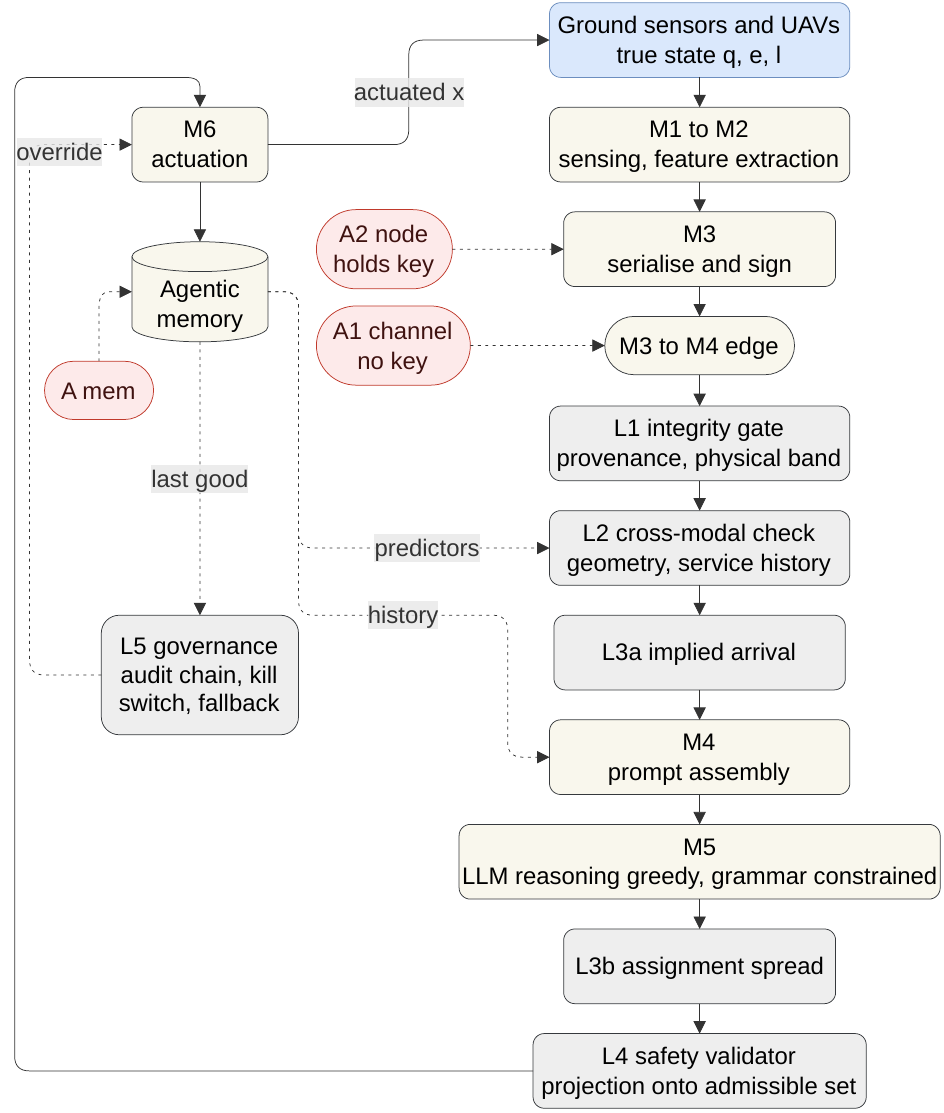}
\caption{Pipeline, attack points, and corresponding defense layers. Input-side layers guard observations before $M_4$, output-side layers guard assignments before $M_6$, and governance spans the slot. $\mathcal{A}_1$ intercepts $M_3\!\rightarrow\!M_4$, $\mathcal{A}_2$ controls $M_3$, and $\mathcal{A}_{\mathrm{mem}}$ poisons agent memory; memory remains independent of observation-side attacks, enabling the predictors in~\eqref{eq:lgeo} and~\eqref{eq:emem}.}

\label{fig:pipeline}
\end{figure}

\subsection{Evaluation Metrics}\label{subsec:metrics}

Let $J^{\mathrm{cl}}$ be the cumulative cost under clean conditions with no defense. The degradation ratio $\Gamma_c=\mathbb{E}[J_c]/\mathbb{E}[J^{\mathrm{cl}}]$ reduces to unity for a fully mitigated attack; the undefended PMA was reported at $14.2\times$, but because the substrate differs we re-evaluate that condition rather than reuse the value. The attack success rate
\begin{equation}\label{eq:asr}
\mathrm{ASR}_c=\frac{1}{T}\sum_{t\in\mathcal{T}}
\mathbb{1}\big[x_c(t)=x^{\mathrm{adv}}(t)\big]
\end{equation}
is the fraction of slots in which the adversary's intended assignment is actuated; unlike $\Gamma$, it separates suppressing control of the schedule from limiting its cost. The full $U\times N$ matrix must match, so partial diversion is priced by $\Gamma$ instead, and the statistic is read against the chance level $N^{-U}=1.25\times10^{-4}$ of a uniform scheduler: the undefended $0.066$ is $528$ times that level and the defended $0.001$ is $8$ times it, so the detectors remove $98.7\%$ of the adversary's control in excess of chance, more under a coarser null model. Detection rates are computed over perturbed slots and false-positive rates over clean runs at the calibration of~\eqref{eq:calib}, and for $\mathcal{A}_2$ also as functions of $\kappa$. We also report the defense tax $\Delta^{\mathrm{tax}}=\mathbb{E}[J^{\mathrm{cl}}_{\mathrm{def}}]-\mathbb{E}[J^{\mathrm{cl}}]$ and the per-layer latency of~\eqref{eq:latency}.

\section{Adversary Model: A Capability Ladder}\label{sec:adversary}

An adversary is the tuple $\mathcal{A}=\langle\mathcal{E},\mathcal{K},\mathcal{V},\mathcal{B}\rangle$ comprising the pipeline elements $\mathcal{E}$ to which it holds read-write access, the key material $\mathcal{K}$ it possesses, the subset $\mathcal{V}\subseteq\{L_1,\dots,L_5\}$ of deployed layers whose specification it knows, and its perturbation budget $\mathcal{B}$. In every class the physical environment, the trajectories, and the model parameters $\theta$ remain unchanged, and only reported values in~\eqref{eq:obs} are altered.

The PMA of LAUS promotes a target set $\mathcal{P}\subset\mathcal{N}$ and suppresses an urgent set $\mathcal{S}\subset\mathcal{N}$, with $\mathcal{P}\cap\mathcal{S}=\emptyset$, by scaling reported features multiplicatively at intensity $\kappa>1$. Writing $\delta_n=+1$ for $n\in\mathcal{P}$, $\delta_n=-1$ for $n\in\mathcal{S}$, and $\delta_n=0$ otherwise, and letting $s_\phi\in\{+1,-1\}$ encode whether an increase in feature $\phi$ raises apparent priority, so $s_q=+1$ while $s_e=s_\ell=-1$, the operator is
\begin{equation}\label{eq:pma}
\tilde\phi_n(t)=\kappa^{\,\delta_n s_\phi}\,\hat\phi_n(t),
\qquad \phi\in\Omega\subseteq\{q,e,\ell\},
\end{equation}
where $\Omega$ is the modified feature subset. The budget is $\mathcal{B}=(\kappa,\Omega,|\mathcal{P}|,|\mathcal{S}|)$, and the configuration of LAUS is recovered at $\kappa=3$ with $\Omega=\{q,e,\ell\}$.
\par
The four classes are not an arbitrary enumeration: each removes exactly the asymmetry the preceding layer exploits, so the ladder terminates once no exploitable redundancy remains.

\emph{$\mathcal{A}_1$, channel adversary}, has $\mathcal{E}=\{M_3\!\rightarrow\!M_4\}$, $\mathcal{K}=\emptyset$, $\mathcal{V}=\emptyset$, $\Omega=\{q,e,\ell\}$. This instantiates the PMA directly: the adversary intercepts the serialized observation in transit but holds no key issued to $M_3$.

\emph{$\mathcal{A}_2$, compromised perception node}, has $\mathcal{E}=\{M_3\}$, $\mathcal{K}=\{k_{M_3}\}$, $\mathcal{V}=\{L_1\}$, $\Omega=\{q,e,\ell\}$. It authenticates the tampered observation legitimately, so detection must rely on physical redundancy.

\emph{$\mathcal{A}_3$, redundancy-aware adversary}, is as $\mathcal{A}_2$ with $\mathcal{V}=\{L_1,L_2\}$ and $\Omega=\{q\}$. The restriction is forced rather than chosen: reported path loss admits an independent prediction from swarm geometry through~\eqref{eq:plos} and~\eqref{eq:pathloss}, reported energy admits one from service history through~\eqref{eq:energy}, and backlog admits neither. Backlog is therefore the only feature a signing adversary can touch without contradicting a source it does not control, while remaining bound by the conservation identity of~\eqref{eq:queue} that $L_3$ enforces.

\emph{$\mathcal{A}_4$, policy-aware adaptive adversary}, retains the access and key material of $\mathcal{A}_2$, has $\mathcal{V}=\{L_1,\dots,L_4\}$ and requires the induced assignment to stay inside the region enforced at $M_6$. Unlike $\mathcal{A}_3$ it takes $\Omega=\{q,\ell\}$ and operates within the interval left open by the propositions, since path loss is detected above $\kappa^{\mathrm{det}}_\ell$ rather than being unusable. Per slot it selects the largest intensity whose residual remains within the assumed threshold and rotates the suppressed sensor it releases so that~\eqref{eq:theta} never binds; energy stays excluded, since~\eqref{eq:ecum} accumulates the drift the one-step check tolerates. It solves
\begin{equation}\label{eq:adaptive}
\max_{\tilde o\in\mathcal{O}_{\mathcal{B}}}\;
\mathbb{E}\big[J(\tilde o)\big]
\;\;\text{s.t.}\;\;
\mathcal{G}_{\theta}\big(\mathrm{asm}(\tilde o,\mathcal{M},g)\big)\in\mathcal{X}_{\mathrm{adm}},
\end{equation}
where $\mathcal{O}_{\mathcal{B}}$ is the set of observations reachable under the budget. The input-side layers then have little residual evidence to exploit, with $L_3$ rejecting too few slots to trigger a response, and the objective shifts from detection to bounding damage. A secondary adversary $\mathcal{A}_{\mathrm{mem}}$ writes a fabricated decision record into the memory~\eqref{eq:memory} rather than into the observation, testing whether defenses calibrated on observation-side manipulation generalize to a different injection point.

\section{Defense-in-Depth Architecture}\label{sec:did}

We replace the block-diagram functions of LAUS with an operational specification in which each layer is a computable predicate over quantities defined above. Each layer is evaluated against the adversary classes that reach it; for $L_1$--$L_4$, layer $L_i$ is associated with the class $\mathcal{A}_i$ that defeats the preceding layer. Thus, $L_4$ and $L_5$ are not detectors: against $\mathcal{A}_4$ they bound damage. Intensity is the capability measure because a bound on $\kappa$ both fixes the layer boundary and quantifies the perturbation magnitude.

\subsection{\texorpdfstring{$L_1$}{L1}: Integrity Gate}\label{subsec:l1}

The gate is the conjunction of two independent tests, which fail against different adversaries,
\begin{equation}\label{eq:l1}
\begin{split}
L_1\big(\sigma(t)\big)={}&\mathbb{1}\big[\mathrm{Vrf}_{k_{M_3}}(\sigma(t))\big]\\
&\times\prod_{n,\phi}\mathbb{1}\big[\hat\phi_n(t)\in[\phi^{\min},\phi^{\max}]\big].
\end{split}
\end{equation}

\subsubsection{Provenance}
Module $M_3$ emits, alongside the observation, a tag
\begin{equation}\label{eq:mac}
\mathrm{tag}(t)=\mathrm{MAC}_{k_{M_3}}\big(\,\mathrm{canon}(\sigma(t))\,\|\,t\,\|\,\mathrm{ctr}(t)\,\big),
\end{equation}
with $\mathrm{canon}(\cdot)$ a canonical byte encoding of~\eqref{eq:obs} and $\mathrm{ctr}(t)$ a strictly monotonic counter. Binding the slot index and counter is necessary, since authenticating the payload alone would let $\mathcal{A}_1$ replay a previously valid observation without forging a tag, so verification at $M_4$ rejects stale or already accepted counters regardless of tag validity. Provenance terminates $\mathcal{A}_1$ and provides no protection against $\mathcal{A}_2$, which signs legitimately; that asymmetry defines the boundary of the first layer.

\subsubsection{Physically Derived Admissibility}
A schema check cannot distinguish manipulation from a legitimate operating condition, since a backlogged sensor on a degraded link produces in-schema values. The bounds here are derived from the channel model instead, so path loss reported from altitude $H$ over a region of side $D$ must satisfy
\begin{equation}\label{eq:rangebounds}
\begin{split}
\ell^{\min}&=20\log_{10}\!\Big(\tfrac{4\pi H}{\lambda}\Big)+\eta_{\mathrm{LoS}},\\
\ell^{\max}&=20\log_{10}\!\Big(\tfrac{4\pi\sqrt{H^{2}+2D^{2}}}{\lambda}\Big)+\eta_{\mathrm{NLoS}},
\end{split}
\end{equation}
obtained by evaluating~\eqref{eq:pathloss} at an overhead sensor under LoS and at the opposite corner under non-LoS; the remaining bounds are $Q_{\max}$ and $E_0$. Since~\eqref{eq:rangebounds} bounds the state while $L_1$ sees noisy measurements, and reported energy starts at exactly $E_0$, the unwidened bounds reject $2.13\%$ of energy and $22.96\%$ of path-loss measurements on clean runs. The bounds for each measured feature are therefore widened by $z_{\mathrm{band}}\sigma_\phi$. We use $z_{\mathrm{band}}=3.5$, giving a path-loss interval of $[71.692,\,98.963]$~dB and zero clean rejections. Backlog is a counter and receives no margin.

\begin{proposition}\label{prop:kappa}
Let $L_1$ admit a reported feature only if it lies in its physically derived band. Then a perturbation of the form~\eqref{eq:pma} at intensity $\kappa$ passes the admissibility test only if $\kappa\le\bar\kappa(t)$, where
\begin{equation}\label{eq:kappabar}
\bar\kappa(t)=\min\left\{
\min_{n\in\mathcal{P}}\left\{\tfrac{Q_{\max}}{\hat q_n(t)},\;
\tfrac{\hat\ell_n(t)}{\ell^{\min}}\right\},\;
\min_{n\in\mathcal{S}}\left\{\tfrac{E_0}{\hat e_n(t)},\;
\tfrac{\ell^{\max}}{\hat\ell_n(t)}\right\}
\right\}.
\end{equation}
The inner minima identify the most restrictive sensor-level constraint within the promoted and suppressed sets and the outer minimum selects the tighter of the two. Two of the six scaled quantities do not appear in either bound, because suppressing backlog on $\mathcal{S}$ and promoting energy on $\mathcal{P}$ both drive the reported value toward zero, and a lower bound at zero imposes no finite constraint when $\kappa\ge1$. An empty set contributes $+\infty$.
\end{proposition}

\begin{proof}
See Appendix~\ref{app:proofs}.
\end{proof}

Two implications follow. First, admissibility provides a manipulation budget whose worst-case value is computable from deployment parameters. The interval from $\ell^{\min}=73.44$~dB to $\ell^{\max}=97.21$~dB gives the worst-case admissible scaling $\bar\kappa\le1.324$, or $1.380$ with the measurement margin, while over clean calibration slots~\eqref{eq:kappabar} has mean $1.019$ and maximum $1.041$, because observed path loss occupies a small part of the interval; the nominal $\kappa=3$ is thus infeasible by more than a factor of two. In the opening slot, where every sensor reports its full energy, the unwidened bound gives $\bar\kappa(0)=1$ exactly and the deployed band $1.0035$. Second, the bound is asymmetric. Promoted backlog and path loss are bounded by $Q_{\max}$ and $\ell^{\min}$, suppressed energy and path loss by $E_0$ and $\ell^{\max}$, but the two directions that drive a value toward zero remain unconstrained: a congested sensor can be made to look idle and a healthy sensor can be made to appear depleted, and both survive $L_1$.

When the gate rejects, the offending fields revert to their values in the last accepted observation, and the substitution is recorded, so the loop never stalls. This stale-hold rule is inherited from LAUS and drives the detection cost of Section~\ref{subsec:results}.

\subsection{\texorpdfstring{$L_2$}{L2}: Cross-Modal Consistency}\label{subsec:l2}

Against $\mathcal{A}_2$ the observation is authentic and in band, so detection must rest on redundancy between modules, and two features each admit an independent predictor that bypasses $M_3$. The geometric predictor evaluates~\eqref{eq:pathloss} directly,
\begin{equation}\label{eq:lgeo}
\ell^{\mathrm{geo}}_n(t)=\min_{u\in\mathcal{U}} L_{n}\big(b_u(t)\big),
\end{equation}
from the trajectory waypoints $b_u(t)$ and the sensor coordinates $c_n$, both fixed at deployment and neither produced by $M_3$, so it is a closed-form function of $t$ that requires no telemetry. The energy predictor is applied over one step rather than integrated from initialization, so error does not accumulate,
\begin{equation}\label{eq:emem}
e^{\mathrm{mem}}_n(t)=\big[\hat e_n(t{-}1)-\epsilon_{\mathrm{tx}}\mu_n(t{-}1)-\epsilon_{\mathrm{idle}}\big]^{+},
\end{equation}
with $\hat e_n(t{-}1)$ and $\mu_n(t{-}1)$ read from memory. Since an adversary could drift reported energy below $\tau_{\mathrm{cmc}}$ each slot and accumulate an arbitrary offset, the layer also evaluates the cumulative residual
\begin{equation}\label{eq:ecum}
r^{e,\mathrm{cum}}_n(t)=\Big|\hat e_n(t)-\hat e_n(t{-}W_m)
+\!\!\sum_{j=1}^{W_m}\!\!\big(\epsilon_{\mathrm{tx}}\mu_n(t{-}j)+\epsilon_{\mathrm{idle}}\big)\Big|.
\end{equation}
The one-step residuals are
\begin{equation}\label{eq:residual}
r^{\ell}_n(t)=\big|\hat\ell_n(t)-\ell^{\mathrm{geo}}_n(t)\big|,
\quad
r^{e}_n(t)=\big|\hat e_n(t)-e^{\mathrm{mem}}_n(t)\big|,
\end{equation}
and the layer accepts only if these and~\eqref{eq:ecum} fall below thresholds obtained from~\eqref{eq:calib} on clean runs. The detection power of $L_2$ against $\mathcal{A}_2$ is then a threshold in $\kappa$ that follows from~\eqref{eq:noise} in closed form.

\begin{proposition}\label{prop:l2detect}
Let $\tau_{\mathrm{cmc}}$ be the $(1-\alpha)$ quantile of the clean per-slot residual. Because the layer accepts on the maximum over $N$ sensors, that statistic is the maximum of $N$ half-normal variates under~\eqref{eq:noise}, so $\tau_{\mathrm{cmc}}=z_N\sigma_\ell$ with $z_N=\Phi_{\mathcal{N}}^{-1}\!\big(\tfrac12\big(1+(1-\alpha)^{1/N}\big)\big)$, where $\Phi_{\mathcal{N}}$ is the standard normal distribution function, distinct from the link-success mapping $\Phi$ of~\eqref{eq:success}, and $z_N$ exceeds the single-scalar quantile $z_{1-\alpha}$. Let the adversary retain $\ell\in\Omega$ at intensity $\kappa$. Then the path-loss residual of~\eqref{eq:residual} exceeds $\tau_{\mathrm{cmc}}$, and $L_2$ rejects, for every
\begin{equation}\label{eq:kappadetect}
\kappa\;\ge\;\kappa^{\mathrm{det}}_\ell=
\max\left\{
\frac{\ell^{\min}}{\ell^{\min}-\tau_{\mathrm{cmc}}},\;
1+\frac{\tau_{\mathrm{cmc}}}{\ell^{\min}}
\right\},
\end{equation}
where the first argument governs the promoted set and the second the suppressed set, both at the worst case $\ell_n=\ell^{\min}$.
\end{proposition}

\begin{proof}
See Appendix~\ref{app:proofs}.
\end{proof}

The boundary depends only on the band and the noise scale, so it is known before any run. The analytical $z_N=3.479$ gives $\tau_{\mathrm{cmc}}=1.740$~dB and $\kappa^{\mathrm{det}}_\ell=1.024$; the deployed threshold, the calibrated quantile of $1.787$~dB plus the $5\%$ margin of Section~\ref{subsec:calibration}, is $1.877$~dB and gives $1.026$, whereas the single-scalar quantile would have given $1.018$, below $\kappa=1.02$, where measured detection is only $0.767$. Together, Proposition~\ref{prop:kappa} (realized mean $1.019$) and Proposition~\ref{prop:l2detect} (all perturbations detected at $\kappa\ge1.026$) confine an adversary that is admissible and undetected to roughly $\kappa\in[1,1.02]$, a window that closes as $\sigma_\ell\rightarrow0$ under the model assumptions. Appendix~\ref{app:mismatch} shows that this boundary is a property of the sensing-noise floor rather than of the multiplicative form of~\eqref{eq:pma}, and quantifies through~\eqref{eq:kappamis} how model mismatch in $\ell^{\mathrm{geo}}_n$ shifts it. Backlog has no independent measurement, which motivates $\mathcal{A}_3$ and a layer exploiting the conservation law governing backlog.

\subsection{\texorpdfstring{$L_3$}{L3}: Memory-Informed Anomaly Detection}\label{subsec:l3}

Against $\mathcal{A}_3$ the observation is authentic, in band, and consistent with every independent predictor. One test remains, and it tests the conservation identity the backlog must satisfy~\cite{mogollon2026physics} rather than a second measurement of it. Arrivals are unobserved but not unconstrained, so inverting the queue recursion over one step gives the arrival count implied by two consecutive reports,
\begin{equation}\label{eq:impliedarrival}
\hat\alpha_n(t{-}1)=\hat q_n(t)-\big[\hat q_n(t{-}1)-\mu_n(t{-}1)\big]^{+},
\end{equation}
which must lie in $[0,\alpha^{\max}]$ for any physically realizable arrival process.

\begin{proposition}\label{prop:kappaq}
Let sensor $n\in\mathcal{P}$ carry a persistent multiplicative backlog perturbation at intensity $\kappa$, so $\hat q_n(t')=\kappa q_n(t')$ throughout the window, and suppose neither clamp in~\eqref{eq:queue} is active, that is $\kappa q_n(t{-}1)\ge\mu_n(t{-}1)$ and $\kappa q_n(t)\le Q_{\max}$. Then
\begin{equation}\label{eq:alphahat}
\hat\alpha_n(t{-}1)=\kappa\,\alpha_n(t{-}1)-(\kappa-1)\,\mu_n(t{-}1),
\end{equation}
and the perturbation passes the admissibility test only if
\begin{equation}\label{eq:kappaq}
\kappa\le\bar\kappa_q(t)=
\begin{cases}
\dfrac{\mu_n(t{-}1)}{\mu_n(t{-}1)-\alpha_n(t{-}1)}, & \mu_n(t{-}1)>0,\\[1.6ex]
\dfrac{\alpha^{\max}}{\alpha_n(t{-}1)}, & \mu_n(t{-}1)=0.
\end{cases}
\end{equation}
\end{proposition}

\begin{proof}
See Appendix~\ref{app:proofs}.
\end{proof}

The served branch is the operative one. Since $\alpha_n\le\alpha^{\max}$, the bound is at most $R/(R-\alpha^{\max})=24/20=1.2$ under Table~\ref{tab:params}, so any intensity above $1.2$ is rejected the first time a promoted sensor is served; at $\kappa=3$ the implied arrival is $3\alpha_n-48\le-36$. This bounds $\mathcal{A}_3$ by design rather than by calibration: draining $\mu$ packets must reduce the reported backlog by $\mu$, whereas an inflated report falls by $\kappa\mu$, leaving $(\kappa-1)\mu$ packets no non-negative arrival can explain, and an adversary that keeps its target unserved to stay hidden forfeits the objective. When the inner clamp is active the test falls back to the unserved branch, and a field exceeding $Q_{\max}$ has already been repaired by the band check. The check costs $O(N)$.

\begin{corollary}\label{cor:noadmissible}
Let arrivals be drawn independently of the service outcome from a distribution clipped at $\alpha^{\max}$, so that $\pi_{\max}=\Pr[\alpha_n=\alpha^{\max}]>0$ and the operator knows $\alpha^{\max}$ exactly, and let a promoted sensor carry a persistent perturbation at any $\kappa>1$. Then in every slot the test rejects with probability at least $\pi_{\max}\Pr[\mu_n=0]$, and with probability one whenever $\mu_n>0$ and $\kappa>R/(R-\alpha^{\max})$. The slots until first rejection are stochastically dominated by a geometric variable whose expectation is finite for every $\kappa>1$.
\end{corollary}

\begin{proof}
See Appendix~\ref{app:proofs}.
\end{proof}

Propositions~\ref{prop:kappa} and~\ref{prop:l2detect} leave $\mathcal{A}_2$ a narrow interval, while Corollary~\ref{cor:noadmissible} leaves $\mathcal{A}_3$ none: every $\kappa>1$ is detected in finite expected time, so a smaller $\kappa$ only delays detection. The premise is a modeling choice: the arrival cap is attained and known in the simulator of~\cite{emami2026laus}, but under an unbounded arrival law or a loose declared cap $\pi_{\max}\rightarrow0$, and an adversary with $\kappa\le1.2$ keeping its target unserved is not detected at all.

A second, far more sensitive mechanism operates at the onset slot $t_0$ alone, where only one of the two reports entering~\eqref{eq:impliedarrival} is scaled, $\hat q_n(t_0)=\kappa q_n(t_0)$ but $\hat q_n(t_0{-}1)=q_n(t_0{-}1)$, so
\begin{equation}\label{eq:onsettransient}
\hat\alpha_n(t_0{-}1)=\kappa\,\alpha_n(t_0{-}1)
+(\kappa-1)\big[q_n(t_0{-}1)-\mu_n(t_0{-}1)\big]^{+},
\end{equation}
whose second term, absent from~\eqref{eq:alphahat}, scales with the standing backlog. Rejection occurs whenever $\kappa>(\alpha^{\max}+\bar q)/(\alpha_n+\bar q)$ with $\bar q=[q_n(t_0{-}1)-\mu_n(t_0{-}1)]^{+}$, which for an unserved sensor with the initial backlog of $21$ packets gives $\kappa>1.06$. An attack present from $t=1$ has no unperturbed predecessor and does not trigger this transient mechanism, so the evaluation reports onset at $t=1$ and $t=\lceil T/3\rceil$ separately.

\subsection{\texorpdfstring{$L_4$}{L4}: Safety Validator}\label{subsec:l4}

The first statistic is the largest share of assignments captured by any $k$ sensors over a window of $W$ slots,
\begin{equation}\label{eq:psi}
\Psi_k(t;W)=\max_{\substack{S\subseteq\mathcal{N}\\|S|=k}}
\frac{1}{UW}\sum_{t'=t-W+1}^{t}\;\sum_{u\in\mathcal{U}}\;\sum_{n\in S}x_{u,n}(t'),
\end{equation}
computed on the actuated pattern, rather than on a reported feature, so an observation-side adversary cannot modify it directly. The denominator $UW$ makes $\Psi_k$ a share of a fixed budget rather than a rate, since a rate would be biased upward before the window fills and would trigger on $32.7\%$ of clean warm-up slots. The pair $(k,W)$ must satisfy
\begin{equation}\label{eq:kwcond}
k<U,
\qquad
W>\max\Big\{\frac{R}{\mathbb{E}[\alpha]},\;T_b\Big\},
\end{equation}
where $T_b$ is the trajectory period in slots: for $k\ge U$ the $k$ most visited sensors cover the swarm in one slot, and a window shorter than either the revisit period $R/\mathbb{E}[\alpha]$ or $T_b$ saturates $\Psi_k$ regardless of the adversary. These periods are $9.2$ and $10$ slots, so $(k,W)=(2,12)$; under a deterministic reference scheduler the clean-attacked separation for $k=2$ rises from $4.0$ pooled standard deviations at $W=8$ to $15.0$ at $W=12$ and $21.9$ at $W=20$.

Concentration bounds the service any $k$ sensors may capture, but not the service denied to any one sensor. The complementary statistic is the longest interval any sensor has gone unserved,
\begin{equation}\label{eq:theta}
\Theta(t;x)=\max_{n\in\mathcal{N}}\big(t-\varsigma_n(t)\big),
\end{equation}
with $\varsigma_n$ as in~\eqref{eq:repair}, evaluated on the proposal so a sensor served now resets to zero and one never visited counts as $t+1$; it needs no window, since $\Theta(t)\le t+1$. Both constraints are needed because the attack redistributes rather than concentrates service, changing the identity of the underserved set while leaving concentration nearly unchanged: over the calibration runs the maximum of $\Psi_2$ is $13/36$ clean and $14/36$ attacked, whereas $\Theta$ exceeds $18$ in no clean slot but in $39.5\%$ of attacked slots, with its maximum increasing from $18$ to $30$, corresponding to a sensor never served. A validator built on~\eqref{eq:psi} alone would miss this attack.

The validator projects the proposal onto the region jointly induced by the two statistics,
\begin{equation}\label{eq:fair}
\mathcal{X}_{\mathrm{adm}}=\big\{x\in\mathcal{X}:
\Psi_k(t;W)\le\beta_k \;\wedge\; \Theta(t;x)\le\theta_{\max}\big\},
\end{equation}
and enforces it by repair rather than rejection, since a rejected assignment still leaves the swarm to be scheduled. While either constraint is violated, the procedure removes one UAV from the sensor that constraint objects to, the arg-max subset $S^{\star}$ attaining $\Psi_k$ in the first case and any served sensor in the second, and reassigns it to
\begin{equation}\label{eq:repair}
\begin{split}
n^{\star}&=\arg\max_{n\in\mathcal{R}_u(t)\setminus S^{\star}}\big(t-\varsigma_n(t)\big),\\
\mathcal{R}_u(t)&=\big\{n:\|c_n-b_u(t)\|\le r_{\mathrm{exc}}\big\},\\
\varsigma_n(t)&=\max\big\{t'<t:\textstyle\sum_u x_{u,n}(t')\ge1\big\}.
\end{split}
\end{equation}
Reachability and time since last visit come from the deployment configuration and the visit history, not from a reported feature, so a compromised observation cannot steer the repair. Ties break by sensor index, keeping the repair consistent with~\eqref{eq:determinism}, and it terminates in at most $U$ iterations with cost $O(N)$.

\subsection{\texorpdfstring{$L_5$}{L5}: Governance, Audit, and Fallback}\label{subsec:l5}

The governance layer performs no detection; it records events and bounds the duration of an undetected compromise. Each slot appends a record $\varrho_t=(t,\mathrm{H}(\sigma(t)),L_{1:4}\text{ verdicts},\tilde x(t),x(t))$ to a hash-chained log,
\begin{equation}\label{eq:hashchain}
h_t=\mathrm{H}\big(h_{t-1}\,\|\,\varrho_t\big),\qquad h_0=\mathrm{H}(\text{nonce}),
\end{equation}
so retrospective alteration of any verdict would invalidate every later digest, which matters when a compromised perception node could otherwise erase the evidence of its own rejections. The kill switch monitors the rejection rate over a sliding window of $W_g$ slots. With the per-slot rejection indicator
\begin{equation}\label{eq:rejflag}
r(t)=\mathbb{1}\Big[\mathrm{flag}_1(t)\vee\mathrm{flag}_2(t)\vee\mathrm{flag}_{3}(t)
\vee\big(\tilde x(t)\notin\mathcal{X}_{\mathrm{adm}}\big)\Big],
\end{equation}
the switch replaces the policy output with the fallback once
\begin{equation}\label{eq:killswitch}
\frac{1}{W_g}\sum_{t'=t-W_g+1}^{t} r(t')>\gamma
\;\Longrightarrow\;
x(t)=\Pi_{\mathrm{fb}}\big(\mathcal{M}(t)\big),
\end{equation}
where the fallback is the deterministic scheduler
\begin{equation}\label{eq:fallback}
\Pi_{\mathrm{fb}}\big(\mathcal{M}(t)\big)\in
\arg\max_{x\in\mathcal{X}}\sum_{u,n}x_{u,n}\,
\mathbb{1}\big[n\in\mathcal{R}_u(t)\big]\,
\big(t-\varsigma_n(t)\big)\,\check q_n,
\end{equation}
with $\check q_n$ the backlog in the most recent observation accepted jointly by $L_1$ and $L_2$. It invokes no language model, so prompt manipulation cannot propagate into it; it does not read the live observation, so it does not inherit the corruption that triggered it, unlike any scheduler consuming $\tilde o(t)$, including the in-context baseline of~\cite{emami2025icldc}; and it schedules only reachable sensors. Once triggered, the switch stays in fallback for at least $W_h$ slots and readmission requires $W_h$ consecutive clean slots, preventing rapid switching near $\gamma$.

\subsection{Composition, Precedence, and Complexity}\label{subsec:composition}

Layers $L_1$ to $L_3$ execute before prompt assembly, $L_4$ after inference, and $L_5$ operates across the slot (Algorithm~\ref{alg:loop}). Two precedence rules are not implied by this ordering and must therefore be specified explicitly. Input-side rejection short-circuits inference when the substituted observation reproduces an already evaluated prompt. The joint $L_1$--$L_2$ verdict is also informative: $L_1$ failing with $L_2$ passing is consistent with a transport fault and is handled locally, whereas $L_1$ passing with $L_2$ failing indicates an authenticated but physically inconsistent observation and is escalated to $L_5$. The per-slot computational cost is $O(UN+N\log N)$, dominated by the cross-modal recomputation of~\eqref{eq:pathloss} over UAV and sensor pairs and the validator's sorting operation, which Appendix~\ref{app:scaling} verifies empirically.

\begin{algorithm}[!tb]
\caption{Defended per-slot control loop}\label{alg:loop}
\begin{algorithmic}[1]
\State \textbf{input:} $\sigma(t),\mathrm{tag}(t),\mathcal{M}(t)$
\If{$\neg\,\mathrm{Vrf}(\sigma(t),\mathrm{tag}(t))$ \textbf{or} band check fails}
  \State repair offending fields from last accepted observation; $\mathrm{flag}_1\gets\textsc{true}$
\EndIf
\State compute $r^{\ell},r^{e},r^{e,\mathrm{cum}}$ by \eqref{eq:residual} and \eqref{eq:ecum}
\State $\mathrm{flag}_2\gets$ any residual exceeds its threshold
\State $\mathrm{flag}_{3}\gets \hat\alpha_n(t{-}1)\notin[0,\alpha^{\max}]$ for some $n$
\If{$\mathrm{flag}_1\wedge\neg\,\mathrm{flag}_2$}
  \State proceed with repaired observation \Comment{channel fault}
\ElsIf{$\mathrm{flag}_2\vee\mathrm{flag}_{3}$}
  \State escalate to $L_5$ \Comment{authenticated, inconsistent}
\EndIf
\State $\pi(t)\gets\mathrm{asm}(\sigma(t),\mathcal{M}(t),g)$; $\;\tilde x(t)\gets\mathcal{G}_\theta(\pi(t))$
\State $\mathrm{flag}_{4}\gets \Psi_k(t;W)>\beta_k \vee \Theta(t;\tilde x(t))>\theta_{\max}$
\State $x(t)\gets\mathrm{flag}_4\;?\;\Pi_{\mathcal{X}_{\mathrm{adm}}}(\tilde x(t))$ by \eqref{eq:repair}$\;:\;\tilde x(t)$
\If{$\frac{1}{W_g}\sum r(t')>\gamma$ by \eqref{eq:rejflag}}
  \State $x(t)\gets\Pi_{\mathrm{fb}}(\mathcal{M}(t))$ by \eqref{eq:fallback}
\EndIf
\State append $\varrho_t$ to the chain \eqref{eq:hashchain}; \textbf{return} $x(t)$
\end{algorithmic}
\end{algorithm}

\subsection{Threshold Calibration and Latency Accounting}\label{subsec:calibration}

Thresholds are calibrated on clean runs rather than tuned against the attack. For a target per-slot false-positive rate $\alpha$,
\begin{equation}\label{eq:calib}
\tau_{\mathrm{cmc}}=Q_{1-\alpha}\big(\mathcal{R}^{\mathrm{clean}}\big),
\end{equation}
where $Q_{1-\alpha}$ is the empirical quantile and $\mathcal{R}^{\mathrm{clean}}$ collects the residuals of~\eqref{eq:residual} from attack-free runs on seeds disjoint from evaluation, so the false-positive rate is a design target rather than a fitted quantity. Because the clean residual is pure sensing noise, the empirical quantile should converge to the analytical $z_N\sigma_\ell$, and it does: $1.787$~dB against $1.740$~dB. A $5\%$ margin gives the deployed $\tau_{\mathrm{cmc}}=1.877$~dB. On the evaluation seeds the per-branch rates are $0.0022$, $0.0089$ and $0.0011$ for the geometric, one-step energy and cumulative residuals, respectively, each at or below $\alpha=0.01$.

The one-step energy branch is evaluated only when the preceding record was accepted by both $L_1$ and $L_2$, and the cumulative branch only when every record in its window was. Otherwise a substituted anchor makes the residual measure service since the substitution, of order $\epsilon_{\mathrm{tx}}R=8.4$~J rather than $\sigma_e=0.05$~J, producing a second mode with $30$ of $31$ occurrences between $8.37$ and $8.57$~J. The guard does not reduce detection, since $L_2$ rejects every perturbed slot of $\mathcal{A}_2$ either way, and lowers the branch false-positive rate from $0.0344$ to $0.0011$ and the layer union from $0.0456$ to $0.0122$; replaying the record identity against the ground-truth state gives zero firings over $290$ slots.

The concentration bound is specified rather than calibrated, since calibration would merely encode the concentration the unconstrained agent produces. At $W=12$, $\Psi_2$ lives on the lattice $j/36$ with maxima $13/36$ clean and $14/36$ attacked, so no budget both preserves clean operation and constrains the attack; we set $\beta_2=15/36=0.4167$, which is why concentration never fires. $\Theta$ attains exactly $18$ over the calibration runs, so $\theta_{\max}=18$ is the tightest budget with no calibration-time activation; it is exceeded in $39.5\%$ of attacked slots and in $1.0\%$ of clean evaluation slots, a clean-condition constraint activation rather than a detector false positive. Alternative starvation budgets trade off attacked and clean repairs: $53.3\%$ against $5.8\%$ at $\theta_{\max}=14$, $46.5\%$ against $0.83\%$ at $16$, and $32.5\%$ against $0.00\%$ at $20$.

With the model co-located, latency has no transport term,
\begin{equation}\label{eq:latency}
\Lambda(t)=\underbrace{\Lambda_{\mathrm{pre}}\big(|\pi(t)|\big)
+\Lambda_{\mathrm{dec}}\big(|\tilde x(t)|\big)}_{\text{inference}}
\;+\;\sum_{i=1}^{5}\Delta_i(t),
\end{equation}
where $\Delta_i$ is the latency added by $L_i$. Each layer performs at most $O(UN+N\log N)$ arithmetic, so the defense overhead is expected to remain small relative to model inference; the measured latency is reported in Section~\ref{sec:case}.
%against about $10^{9}$ operations per inference step, so $\sum_i\Delta_i\ll\Lambda_{\mathrm{pre}}+\Lambda_{\mathrm{dec}}$ on any plausible substrate.

\section{Performance Evaluation Across the Adversary Ladder}\label{sec:case}

\subsection{Testbed, Substrate, and Protocol}\label{subsec:testbed}

Three UAVs serve twenty ground sensors over a $100\times100$~m$^2$ region at $100$~m altitude for thirty slots under the objective of~\eqref{eq:objective} (Table~\ref{tab:params}). The scenario, mission parameters and cost metric are inherited from the LAUS testbed, while the undefended attacked condition is re-run on the present substrate, a locally hosted open-weight model under greedy, grammar-constrained decoding, so absolute costs are not comparable with LAUS: the studies share the scenario, but not the scheduler. An ablation shows that the decision depends on the generated reasoning trace rather than only on the observation: over ten prompts, removing the reasoning trace changes the actuated assignment in every case, and truncation to one-quarter, one-half, and three-quarters changes it in $70\%$, $60\%$ and $20\%$ of cases.

Every condition is evaluated over $N_{\mathrm{run}}=30$ paired seeds, each with its exogenous stream drawn once before any action and then indexed rather than consumed. Seeds run in lockstep, with a slot's prompts dispatched as one batch of fixed width, and decoded assignments are memoized on a hash of the prompt, making $\pi\mapsto\tilde x$ single-valued across conditions; replaying a run from the store reproduced its cost and decision trace exactly. Thresholds are fixed beforehand on disjoint calibration seeds. Each layer is evaluated only against classes that reach it: for $\mathcal{A}_1$, no defense, $L_1$, $L_1$ with the fallback, the validator, and the validator with provenance; for $\mathcal{A}_2$, $L_1$, $L_1$ with $L_2$, that pair with the fallback, and the full stack; for $\mathcal{A}_3$, both configurations reported in Table~\ref{tab:results}; and for the two variants of $\mathcal{A}_4$ and $\mathcal{A}_{\mathrm{mem}}$, the full stack. Clean undefended and defended conditions bracket the grid. Significance uses the Wilcoxon signed-rank test over paired seeds, while the Friedman test assesses differences across all conditions, with percentile bootstrap intervals over $10^4$ resamples, since the attacked cost distribution is asymmetric. The Friedman test over all fifteen conditions gives $\chi^2=353.5$ on $14$ degrees of freedom ($p=7.6\times10^{-67}$). The value $p=1.7\times10^{-6}$ is the two-sided floor of the normal approximation at $n=30$, reached when all thirty paired differences share a sign.
\begin{table}[!tb]
\centering
\caption{Key testbed configuration parameters.}
\label{tab:params}
\footnotesize
\setlength{\tabcolsep}{3pt}
\begin{tabular}{@{}p{0.47\linewidth}p{0.47\linewidth}@{}}
\hline
Parameter & Value \\
\hline
UAVs $U$ / sensors $N$ / horizon $T$ & 3 / 20 / 30 \\
Region $D\times D$ / altitude $H$ & $100\times100$~m$^2$, 100~m \\
Buffer $Q_{\max}$ / energy budget $E_0$ & 60 packets / 50~J \\
Service volume $R$ per visit & 24 packets \\
Arrival mean / cap $\alpha^{\max}$ & 2.6 / 4 packets per slot \\
Attack intensity $\kappa$ & 3 \\
\hline
Model / revision & Qwen3-8B-AWQ / \texttt{4da05a8e} \\
Quantization / decoding & 4-bit AWQ / greedy, grammar-constrained \\
Slot $\Delta t$ / speed limit $v_{\max}$ & 5.0~s / 14.0~m/s \\
Memory window $W_m$ & 5 records \\
Sensing noise $\sigma_\ell$ / $\sigma_e$ & 0.5~dB / 0.05~J \\
Starvation budget $\theta_{\max}$ & 18 slots \\
Kill switch $(\gamma,W_g,W_h)$ & $(0.6,5,3)$ \\
Calibration target $\alpha$ & 0.01 \\
Evaluation / calibration seeds & 0--29 / 100--119 \\
Hardware / serving & RTX 4070, i9-14900K; vLLM 0.27.1 \\
\hline
\end{tabular}
\end{table}

\subsection{Results}\label{subsec:results}

Table~\ref{tab:results} reports the results. Its two central metrics move independently: $\Gamma$ says how much the mission suffered, the success rate whether the adversary was still choosing the schedule.

\begin{table*}[!tb]
\centering
\caption{Per-layer evaluation across the adversary ladder, $N_{\mathrm{run}}=30$. The ratio $\Gamma$ is normalized to the clean undefended row, ASR is~\eqref{eq:asr}, and $p$ is the two-sided Wilcoxon signed-rank value against the undefended attacked row. The 95\% CIs are percentile-bootstrap intervals, and $\sigma$ is the sample standard deviation.}
\label{tab:results}
\footnotesize
\begin{tabular}{@{}llccccccc@{}}
\hline
Condition & Active layers & Mean $J$ & $\sigma$ & 95\% CI & $\Gamma$ & ASR & Fallback (\%) & $p$ \\
\hline
Clean, undefended & --- & 39.93 & 5.86 & [37.93, 42.17] & 1.00 & --- & 0.0 & --- \\
$\mathcal{A}_1$, undefended & --- & 124.97 & 19.69 & [117.80, 131.93]& 3.13 & 0.066 & 0.0 & --- \\
\hline
$\mathcal{A}_1$ & $L_1$ & 223.97 & 67.24 & [200.63, 248.23]& 5.61 & 0.001 & 0.0 & $5.2{\times}10^{-6}$ \\
$\mathcal{A}_1$ & $L_1,L_5$ & 89.37 & 32.78 & [78.23, 101.30] & 2.24 & 0.002 & 93.3 & $1.9{\times}10^{-4}$ \\
$\mathcal{A}_1$ & $L_4$ & 93.10 & 26.68 & [83.80, 102.73] & 2.33 & 0.069 & 0.0 & $1.5{\times}10^{-5}$ \\
$\mathcal{A}_1$ & $L_1,L_4$ & 165.50 & 53.11 & [146.83, 184.10]& 4.14 & 0.000 & 0.0 & $1.0{\times}10^{-3}$ \\
$\mathcal{A}_2$ & $L_1$ & 126.37 & 24.28 & [117.70, 134.77]& 3.16 & 0.070 & 0.0 & $0.81$ \\
$\mathcal{A}_2$ & $L_1,L_2$ & 220.07 & 70.60 & [195.27, 245.57]& 5.51 & 0.002 & 0.0 & $7.0{\times}10^{-6}$ \\
$\mathcal{A}_2$ & $L_1,L_2,L_5$ & 89.63 & 32.53 & [78.60, 101.50] & 2.24 & 0.002 & 93.6 & $2.0{\times}10^{-4}$ \\
$\mathcal{A}_2$ & $L_1$ to $L_5$ & 89.63 & 32.53 & [78.60, 101.50] & 2.24 & 0.002 & 93.6 & $2.0{\times}10^{-4}$ \\
$\mathcal{A}_3$ & $L_1,L_2$ & 38.90 & 7.01 & [36.63, 41.53] & 0.97 & 0.001 & 0.0 & $1.7{\times}10^{-6}$ \\
$\mathcal{A}_3$ & $L_1$ to $L_5$ & 40.87 & 11.82 & [37.50, 45.67] & 1.02 & 0.001 & 1.8 & $1.7{\times}10^{-6}$ \\
$\mathcal{A}_4$ surrogate & $L_1$ to $L_5$ & 40.70 & 6.00 & [38.63, 42.87] & 1.02 & 0.000 & 0.8 & $1.7{\times}10^{-6}$ \\
$\mathcal{A}_4$ policy-exact & $L_1$ to $L_5$ & 43.23 & 15.95 & [38.33, 49.40] & 1.08 & 0.001 & 0.7 & $1.7{\times}10^{-6}$ \\
$\mathcal{A}_{\mathrm{mem}}$ & $L_1$ to $L_5$ & 38.00 & 5.75 & [36.03, 40.10] & 0.95 & 0.016 & 0.3 & $1.7{\times}10^{-6}$ \\
\hline
Clean, defended & $L_1$ to $L_5$ & 41.00 & 7.30 & [38.60, 43.73] & 1.03 & --- & 0.3 & $1.7{\times}10^{-6}$ \\
\hline
\end{tabular}
\end{table*}

\begin{figure*}[!tb]
\centering
\includegraphics[width=0.88\textwidth]{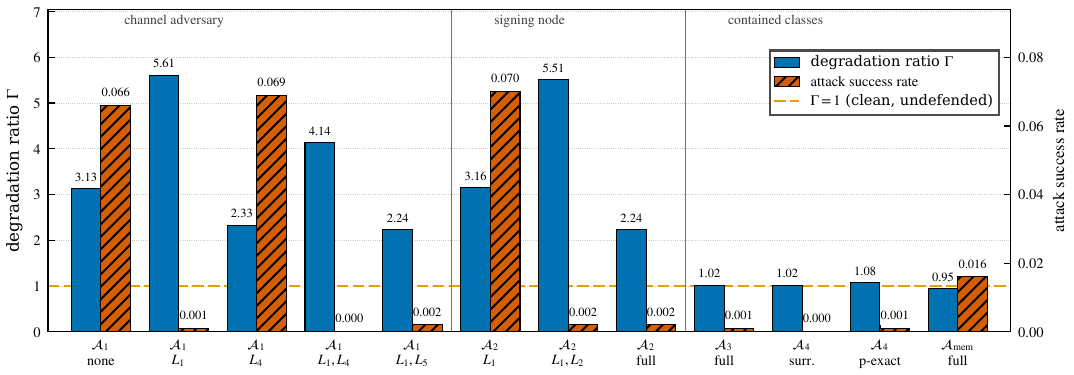}
\caption{Detection and response, from Table~\ref{tab:results} over $30$ paired seeds, grouped by adversary class. Within each group, adding a detector drives the attack success rate to the floor while the degradation ratio rises above the undefended baseline. The validator $L_4$, which detects nothing, is the only single layer that lowers $\Gamma$ against $\mathcal{A}_1$. The two $\mathcal{A}_4$ bars are the surrogate and policy-exact adversaries, differing by $0.06$ in $\Gamma$ after $6{,}960$ queries to the deployed model.}
\label{fig:detectvsrespond}
\end{figure*}

Fig.~\ref{fig:detectvsrespond} places the two metrics side by side, and the $L_4$-only condition shows the distinction most clearly: the validator detects no attack, yet removes $37.5\%$ of the excess cost $\mathbb{E}[J]-\mathbb{E}[J^{\mathrm{cl}}]$, from $85.04$ to $53.17$ packets, so $\Gamma$ falls from $3.13$ to $2.33$. Over the $900$ evaluation slots, concentration is never violated, while starvation is violated in $59$, so every repair comes from starvation. Adding provenance drives the success rate to zero but raises $\Gamma$ from $2.33$ to $4.14$, significantly above the undefended $3.13$ ($p=1.0\times10^{-3}$), because the stale snapshot costs more than the starvation constraint recovers; relative to $L_1$ alone ($\Gamma=5.61$), the validator still removes $31.8\%$ of the excess. The rows at $0.97$ and $0.95$ do not mean the attack improved the mission: $\Gamma$ is normalized to the clean undefended condition, and relative to the clean defended cost of $41.00$ both remain close to unity.

\subsubsection{Rung one, the channel adversary}
The undefended condition re-establishes the PMA vulnerability on this substrate, and provenance terminates $\mathcal{A}_1$, since without $k_{M_3}$ the adversary cannot recompute a tag and the counter binding rejects a replay. The detection rate is unity on every perturbed slot, so what matters is $\Gamma$ with and without $L_5$. Provenance drives the success rate from $0.066$ to $0.001$ while $\Gamma$ rises from $3.13$ to $5.61$ ($p=5.2\times10^{-6}$): the adversary no longer chooses the schedule, yet mission cost is higher than with no defense. The fallback returns $\Gamma$ to $2.24$ at the price of $93.3\%$ of slots under fallback control. The pattern recurs at rung two, where cross-modal consistency drives the success rate from $0.070$ to $0.002$ while $\Gamma$ rises from $3.16$ to $5.51$ ($p=7.0\times10^{-6}$).

\subsubsection{Rung two, the compromised perception node}
Because $\mathcal{A}_2$ signs its own observation, only the admissibility component of $L_1$ can act, and these rows test the asymmetry of Proposition~\ref{prop:kappa}: four of the six scaled quantities are held to a realized intensity of about $1.02$, while the other two remain unconstrained. At $\kappa=3$ the promoted backlog is scaled by $3.00$ and the promoted energy by $0.33$, both unclipped, while the suppressed energy is bounded at $1.00$ and the promoted path loss at $0.97$. Damage propagates through the unconstrained directions, so $L_1$ alone is statistically indistinguishable from no defense ($\Gamma=3.16$ against $3.13$, $p=0.81$). With reported path loss pinned at the band floor, the residual $\ell_n-\ell^{\min}$ persists, at every admissible intensity, which is why the true-positive rate stays at unity. Once $L_2$ is active the terminal layers are redundant for this class: the full stack and $L_1$, $L_2$, $L_5$ give the same cost on every seed.

\subsubsection{Rung three, the redundancy-aware adversary}
Confined to backlog, $\mathcal{A}_3$ passes $L_2$, and the manipulation the implied-arrival and output-side constraints still permit is not operationally significant. The $L_3$ adversary-caused rejection rate is $0.033$ at a realized promoted intensity of only $1.0055$, because the unserved branch $\alpha^{\max}/\alpha_n$ of~\eqref{eq:kappaq} admits no manipulation in the $26\%$ of slots where arrivals reach the cap. Consistent with Corollary~\ref{cor:noadmissible}, the adversary gains stealth only by reducing the intensity until it produces no measurable damage, so this row confirms the bounds rather than testing them independently (Section~\ref{subsec:threats}).

\subsubsection{Rung four, the policy-aware adaptive adversary}
Since $\mathcal{A}_4$ solves~\eqref{eq:adaptive}, the figure of merit is the residual $\Gamma$ against the undefended value. It is limited by the admissible observation set rather than by~\eqref{eq:fair}, which acts on the actuated pattern, since visits to lightly loaded sensors can let heavily loaded queues grow while the pattern stays dispersed. Propositions~\ref{prop:kappa} and~\ref{prop:l2detect} confine $\mathcal{A}_4$ to $\kappa\in[1,1.02]$, about $1.4$~dB of path-loss distortion, below the $1.74$~dB at which $L_2$ reacts.

Because the model is open-weight, the policy-exact adversary runs a local copy of $\mathcal{G}_\theta$, decoding $K=8$ candidate observations per slot over the open window and committing the one whose induced assignment has the highest expected one-slot cost, requiring $6{,}960$ decodes and $8.4$ GPU hours. The difference in mean cost is not statistically significant ($p=0.99$ against the surrogate); the effect is in the tail. On $14$ of $30$ seeds it does strictly less damage, because greedy per-slot optimization inside so narrow a window can raise cost later, while two seeds add $+47$ and $+57$ packets, raising $\sigma$ from $6.00$ to $15.95$ and $\Gamma$ only to $1.08$. On both outliers the validator performs no repair and $\Theta$ stays within budget: this is residual capability that remains admissible, rather than capability excluded by the bounds. The governance rule is itself exploitable, since $L_3$ rejects fewer than one perturbed slot in ten while the switch triggers only above $\gamma=0.6$ over $W_g=5$ slots, so fallback occupancy stays at $0.7\%$.

\subsubsection{A different injection point and attack onset}
The memory-poisoning variant leaves the live observation clean and consistent with both predictors, so no observation-side layer fires and the case is handled on the decision side. The success rate is $0.016$, the highest among the full-stack defended adversarial conditions and roughly an order of magnitude above the observation-side classes, yet $\Gamma=0.95$ and fallback occupancy is $0.3\%$, as in the clean defended condition: the adversary acquires a small share of schedule control without increasing mission cost. The present runs do not distinguish whether the record shifts the choice only among assignments already satisfying~\eqref{eq:fair} or the starvation constraint repairs the difference; they establish that the decision-side constraints bound damage from an injection point they were not calibrated against, and that no layer identifies it. For attack onset (Appendix~\ref{app:onset}), $\mathcal{A}_3$ gives the same outcome at both onsets, while for $\mathcal{A}_2$ they differ in duration rather than detection, with fallback occupancy falling from $93.6\%$ to $57.1\%$ and $\Gamma$ from $2.24$ to $1.48$.

\subsection{Detection Curves and Per-Layer Rates}\label{subsec:curves}

Proposition~\ref{prop:l2detect} predicts the transition, since a promoted sensor's residual $r^{\ell}_n=\ell_n(1-\kappa^{-1})$ crosses $\tau_{\mathrm{cmc}}$ at $\kappa^{\mathrm{det}}_\ell$, and Fig.~\ref{fig:detection} tests it. For the onset condition, detection reaches unity at the predicted boundary and rises gradually below it, from $0.100$ at $\kappa=1.01$ to $0.767$ at $\kappa=1.02$, because~\eqref{eq:kappadetect} is the worst case for a sensor at the band floor and sensors with larger path loss, and noise, cross earlier; the lower run-average rate reflects slots already operating from a substituted observation. The $L_1$ band rejects $\mathcal{A}_2$ at none of the tested intensities, including $\kappa=3$, as Proposition~\ref{prop:kappa} implies: knowing $L_1$, a rational adversary clips to the band, and once $\kappa>\bar\kappa$ the reported path loss is pinned at the floor and the residual saturates at $\ell_n-\ell^{\min}$. The $L_3$ sweep (Appendix~\ref{app:l3}) shows an onset rejection rate rising from $0.60$ at $\kappa=1.05$ to unity by $\kappa=1.15$, bracketing the transient threshold of $1.06$, and a steady-state rate of at least $0.586$.

\begin{figure}[!tb]
\centering
\includegraphics[width=0.9\columnwidth]{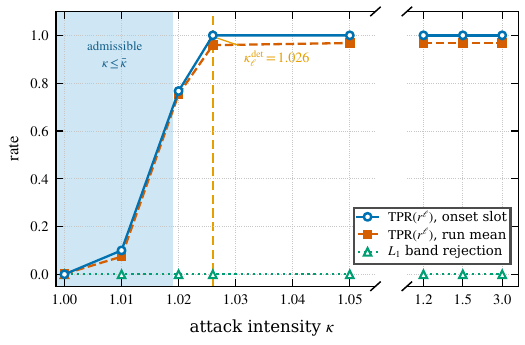}
\caption{Detection of $\mathcal{A}_2$ by $L_2$ against the boundary $\kappa^{\mathrm{det}}_\ell$ that Proposition~\ref{prop:l2detect} fixes before any run. The shaded region is the interval Proposition~\ref{prop:kappa} leaves admissible on the realized geometry. The axis is broken because all rates saturate beyond $\kappa=1.05$.}
\label{fig:detection}
\end{figure}

Provenance rejects every perturbed slot of $\mathcal{A}_1$ at a false-positive rate of $0.001$, the $L_1$ band and the concentration constraint never fire, and $L_4$ starvation fires on $0.068$ of perturbed slots at a clean rate of $0.010$. For the two detectors the aggregate rate conflates a detection, where the current report violates the admissibility condition, with a cascade, where the reference state is itself a substituted observation. Although $L_2$ rejects every perturbed slot of $\mathcal{A}_2$, only about one in fifteen rejections is attributable to the current manipulation, because the last accepted report never refreshes under a persistent adversary. This is how a low adversary-caused rejection rate coexists with a success rate of $0.002$, while $\Gamma=5.51$ quantifies the cost of the resulting freeze; the adversary-caused rate of $L_3$ is likewise only $0.033$ (Table~\ref{tab:onset}).

\subsection{Cost of the Architecture and Sensitivity}\label{subsec:cost}

The defense tax is $\Delta^{\mathrm{tax}}=1.07$ packets, $2.7\%$ of the clean cost, of which $0.64$ packets are attributable to $L_4$ ($L_4$ alone raises the clean cost from $39.93$ to $40.57$), so the layer that removes $37.5\%$ of the attack-induced excess cost accounts for most of the tax, and the stack costs about $1.25\%$ of the $85.04$ packets the attack adds. Table~\ref{tab:cost} decomposes~\eqref{eq:latency}: the stack adds $0.1255$~ms, $0.0025\%$ of the $5$~s period, with $L_2$ contributing the largest latency because it recomputes~\eqref{eq:pathloss} for every UAV and sensor pair. From $N=20$ to $N=1000$ (Appendix~\ref{app:scaling}) the stack grows by $64.2\times$ compared with the $115.3\times$ growth of $N\log N$ and still consumes only $0.06\%$ of the period. The audit chain of~\eqref{eq:hashchain} matched the stored digests in all $180$ runs deploying $L_5$.

\begin{table}[!tb]
\centering
\caption{Per-layer added latency from~\eqref{eq:latency}, averaged over every five-layer condition of Table~\ref{tab:results}. Percentages are of the $5$~s control period.}
\label{tab:cost}
\footnotesize
\begin{tabular}{@{}lccc@{}}
\hline
Component & Mean (ms) & \% of $\Delta t$ & Complexity \\
\hline
$L_1$ integrity gate & 0.0315 & 0.00063 & $O(N+|\sigma|)$ \\
$L_2$ cross-modal & 0.0590 & 0.00118 & $O(UN)$ \\
$L_3$ anomaly & 0.0051 & 0.00010 & $O(N)$ \\
$L_4$ validator & 0.0287 & 0.00057 & $O(N\log N)$ \\
$L_5$ governance & 0.0013 & 0.00003 & $O(1)$ \\
\hline
Stack total & 0.1255 & 0.00251 & $O(UN + N\log N)$ \\
Control period $\Delta t$ & 5000 & 100 & --- \\
\hline
\end{tabular}
\end{table}

A key parameter is the knee of the link-success mapping, which determines whether the system is queue limited or link limited. At the nominal $0.62$, $58.5\%$ of assignments succeed, and sweeping the knee (Appendix~\ref{app:scaling}) raises $\Gamma$ monotonically with link quality from $6.4$ to $52.5$, showing that the reported conclusions persist over a wide range of link conditions.

\subsection{Limitations}\label{subsec:threats}

First, \eqref{eq:pma} is multiplicative in the reported representation, which for path loss is expressed in decibels; the cross-modal boundary is unchanged at $\tau_{\mathrm{cmc}}=1.877$~dB in either parameterization (Appendix~\ref{app:mismatch}), but $\bar\kappa$ is a ratio of logarithmic quantities and is a budget on the reported field rather than a physical constraint. Second, only stale-hold repair is tested, so the finding that detection without response is worse than no defense is established for that rule, not for detection in general. Since the fallback of~\eqref{eq:fallback} discards the observation at $\Gamma=2.24$, any repair at least as good as discarding dominates stale-hold ($\Gamma=5.61$) by a factor of $2.50$; imputation from $\ell^{\mathrm{geo}}_n$ is the obvious candidate, since a frozen path loss is wrong by up to $3.13$~dB against $\sigma_\ell=0.5$~dB, and remains untested. Third, $\mathcal{A}_3$ clips its intensity to the open window, so its null result confirms the bounds rather than testing them; the damage curve against $\kappa$ is not reported, and the fallback is never run on clean observations over a full horizon, so it is not compared against a deterministic heuristic. Fourth, the realized $\bar\kappa=1.019$ reflects a compact geometry: at $D=2$~km, \eqref{eq:rangebounds} gives $\ell^{\max}=121.48$~dB and a worst-case ratio of $1.654$, and~\eqref{eq:kappamis} prices predictor mismatch. Fifth, reproducibility is enforced at the harness rather than the decoder: one prompt decoded in two batches with different neighbors agreed for $1573$ tokens before diverging, so a reproduction from an empty memoization store is reproducible in distribution rather than bit for bit. Finally, the onset transient of~\eqref{eq:onsettransient} cannot see a manipulation present from the first slot, Corollary~\ref{cor:noadmissible} needs an arrival cap both attained and known, and the reasoning budget is reached in every generation, so a different budget is a different policy.

\section{Conclusion}\label{sec:conclusion}

This work implemented and evaluated a five-layer DID architecture to secure the perception–reasoning interface of an LLM-centric UAV scheduler. Three closed-form bounds were derived to characterize admissible manipulation intensity, with the cross-modal boundary $\kappa^{\mathrm{det}}_\ell = 1.026$ matching the measured transition to unit detection. The evaluation shows that detection alone does not guarantee effective protection: provenance and cross-modal checks drive the attack success rate to near zero but can increase degradation when rejected observations are replaced by stale snapshots. The memory-driven fallback and starvation constraint provide recovery, while the full stack defense limits the degradation ratio to $1.02$ against the redundancy-aware adversary and $1.08$ against the policy-aware adversary. These results demonstrate the importance of evaluating detection and response separately when designing layered defenses for LLM-assisted UAV control.

\appendices

\section{Proofs of the Three Propositions and the Corollary}\label{app:proofs}
\emph{Proposition~\ref{prop:kappa}.} Substituting~\eqref{eq:pma} into the band constraint per feature and direction gives, for $n\in\mathcal{P}$, the conditions $\kappa\hat q_n\le Q_{\max}$ and $\hat\ell_n/\kappa\ge\ell^{\min}$, since $\delta_n s_q=+1$ scales backlog upward while $\delta_n s_\ell=-1$ scales path loss downward. For $n\in\mathcal{S}$ it gives $\kappa\hat e_n\le E_0$ and $\kappa\hat\ell_n\le\ell^{\max}$, since $\delta_n s_e$ and $\delta_n s_\ell$ both equal $+1$ there. The two remaining scaled features are driven downward and satisfy $\hat\phi_n/\kappa\ge0$ for every $\kappa\ge1$, so they impose no constraint. Taking the tightest condition yields~\eqref{eq:kappabar}. \hfill$\blacksquare$

\emph{Proposition~\ref{prop:l2detect}.} For $n\in\mathcal{P}$ the exponent $\delta_n s_\ell$ is $-1$, so $\hat\ell_n=\ell_n/\kappa$, and since~\eqref{eq:lgeo} recovers $\ell_n$ up to the noise of~\eqref{eq:noise} the residual is $r^{\ell}_n=\ell_n(1-\kappa^{-1})$. Requiring $r^{\ell}_n>\tau_{\mathrm{cmc}}$ gives $\kappa>\ell_n/(\ell_n-\tau_{\mathrm{cmc}})$, which decreases in $\ell_n$ and is worst at $\ell^{\min}$. For $n\in\mathcal{S}$ the exponent is $+1$, so $\hat\ell_n=\kappa\ell_n$, the residual is $\ell_n(\kappa-1)$, and the same requirement gives $\kappa>1+\tau_{\mathrm{cmc}}/\ell_n$, which also decreases in $\ell_n$ and is worst at $\ell^{\min}$. The larger of the two worst cases yields~\eqref{eq:kappadetect}; since $\ell^{\min}/(\ell^{\min}-\tau_{\mathrm{cmc}})=1+\tau_{\mathrm{cmc}}/(\ell^{\min}-\tau_{\mathrm{cmc}})>1+\tau_{\mathrm{cmc}}/\ell^{\min}$, the promoted-set term always binds, consistent with~\eqref{eq:kappamis}. \hfill$\blacksquare$

\emph{Proposition~\ref{prop:kappaq}.} Away from both clamps, \eqref{eq:queue} reads $q_n(t)=q_n(t{-}1)-\mu_n(t{-}1)+\alpha_n(t{-}1)$. Multiplying by $\kappa$ and substituting into~\eqref{eq:impliedarrival} gives $\hat\alpha_n(t{-}1)=\kappa q_n(t)-(\kappa q_n(t{-}1)-\mu_n(t{-}1))=\kappa\alpha_n(t{-}1)-(\kappa-1)\mu_n(t{-}1)$, which is~\eqref{eq:alphahat}. For $\mu_n(t{-}1)>0$ this decreases strictly in $\kappa$ and becomes negative, hence inadmissible, exactly when $\kappa(\mu_n(t{-}1)-\alpha_n(t{-}1))>\mu_n(t{-}1)$. For $\mu_n(t{-}1)=0$ it reduces to $\kappa\alpha_n(t{-}1)$, which exceeds $\alpha^{\max}$ when $\kappa>\alpha^{\max}/\alpha_n(t{-}1)$. The binding case gives~\eqref{eq:kappaq}. \hfill$\blacksquare$

\emph{Corollary~\ref{cor:noadmissible}.} If $\mu_n(t{-}1)=0$ then~\eqref{eq:alphahat} reduces to $\hat\alpha_n=\kappa\alpha_n$, which for $\alpha_n=\alpha^{\max}$ exceeds $\alpha^{\max}$ at any $\kappa>1$, and clipping makes that draw occur with probability $\pi_{\max}$. If $\mu_n(t{-}1)=R>0$, the served branch of~\eqref{eq:kappaq} rejects whenever $\kappa>R/(R-\alpha^{\max})$. The cases are exhaustive, so the per-slot rejection probability is bounded below as stated and the geometric bound follows. The argument uses $\pi_{\max}>0$ and an exactly known $\alpha^{\max}$; both are premises, not consequences. \hfill$\blacksquare$

\section{Parameterization and Predictor Mismatch in \texorpdfstring{$L_2$}{L2}}\label{app:mismatch}
Because~\eqref{eq:pma} scales a logarithmic quantity, it might be read as producing the tightness of Proposition~\ref{prop:l2detect} in Section~\ref{subsec:l2} by itself. It does not. Scaling the decibel value by $\kappa$ raises the linear path loss to the power $\kappa$, whereas an additive offset $\Delta$ in decibels, the physically natural form, scales it by $10^{\Delta/10}$. For such an offset the residual of~\eqref{eq:residual} is $|\Delta|$ up to sensing noise, so $L_2$ rejects every $|\Delta|>\tau_{\mathrm{cmc}}=1.877$~dB, at most a factor of $1.54$ on the linear scale. The multiplicative boundary expresses the same threshold, since $\ell^{\min}(1-1/\kappa^{\mathrm{det}}_\ell)=\tau_{\mathrm{cmc}}$ exactly by~\eqref{eq:kappadetect}, so the cross-modal boundary is a property of the sensing-noise floor. The admissibility bound $\bar\kappa$ of Proposition~\ref{prop:kappa} is different: a ratio such as $\ell^{\max}/\ell^{\min}$ changes with the reference level of the decibel scale, so $\bar\kappa$ has no physical meaning on its own and applies only to the multiplicative attack model of~\eqref{eq:pma}. An additive adversary is bounded instead by the band slack in decibels and by $\tau_{\mathrm{cmc}}$, whichever is tighter.

The boundary is also tightest here, because $\ell^{\mathrm{geo}}_n$ reproduces the true path loss up to sensing noise. Under a predictor with model mismatch of standard deviation $\sigma_{\mathrm{mis}}$, independent of the sensing noise, the residual scale becomes $\sigma_{\mathrm{eff}}=\sqrt{\sigma_\ell^{2}+\sigma_{\mathrm{mis}}^{2}}$ and Proposition~\ref{prop:l2detect} applies unchanged with $\tau_{\mathrm{cmc}}=z_N\sigma_{\mathrm{eff}}$, so
\begin{equation}\label{eq:kappamis}
\kappa^{\mathrm{det}}_\ell(\sigma_{\mathrm{mis}})=\frac{\ell^{\min}}{\ell^{\min}-z_N\sqrt{\sigma_\ell^{2}+\sigma_{\mathrm{mis}}^{2}}},
\end{equation}
which gives $\kappa^{\mathrm{det}}_\ell=1.024$, $1.056$, $1.108$ and $1.236$ at $\sigma_{\mathrm{mis}}=0$, $1$, $2$ and $4$~dB, an admissible offset of $1.74$, $3.89$, $7.17$ and $14.02$~dB. This does not by itself widen the adversary's window in this deployment, because the constraints compose as $\min\{\bar\kappa,\kappa^{\mathrm{det}}_\ell\}$ and the band is tighter: a realized $\bar\kappa=1.019$ admits $\ell^{\min}(\bar\kappa-1)=1.40$~dB of downward distortion, below the $1.74$~dB at which $L_2$ reacts with an exact predictor, and $L_2$ could bind only for $\sigma_{\mathrm{eff}}<0.40$~dB, which $\sigma_\ell$ alone exceeds. The ordering reverses once the band is loose: at $D=2$~km the band admits $48.0$~dB of slack, so $L_2$ binds for every $\sigma_{\mathrm{mis}}<13.8$~dB and the window is $z_N\sigma_{\mathrm{eff}}$ directly.

\section{Sensitivity to Attack Onset}\label{app:onset}
The implied-arrival test of~\eqref{eq:impliedarrival} and the inter-slot component of $L_1$ compare consecutive reports, so an attack beginning mid-mission presents a discontinuity that one present from the first slot does not. Table~\ref{tab:onset} reports both onsets for the two classes that reach the full stack. For $\mathcal{A}_3$ the rows are similar, which follows from Corollary~\ref{cor:noadmissible} rather than from an insensitive test: \eqref{eq:onsettransient} places the transient threshold at $\kappa>1.06$, whereas the realized intensity on the promoted set is $1.0055$. For $\mathcal{A}_2$ the rows differ in duration rather than detection: the cross-modal layer rejects every attacked slot at both onsets, so the observation is substituted before $L_3$ can act, and occupancy tracks the number of attacked slots, $29$ and $20$ out of $30$, once the dwell time needed to trigger the switch is included. The cost of defending against $\mathcal{A}_2$ is thus close to proportional to the duration of the compromise, the quantity the governance layer of Section~\ref{subsec:l5} is designed to bound.

\begin{table}[!tb]
\centering
\caption{Sensitivity to attack onset under the full stack, $30$ paired seeds per row, with rejection and detection measured on attacked slots only. The $L_3$ column reports the adversary-caused rate; the $L_2$ column reports the total rate, since it is the substitution rather than its cause that drives fallback occupancy.}
\label{tab:onset}
\footnotesize
\begin{tabular}{@{}llccccc@{}}
\hline
Adversary & Onset & $\Gamma$ & $L_2$ rej. & $\mathrm{TPR}(L_3)$ & ASR & Fallback \\
\hline
$\mathcal{A}_3$ & $t=1$ & 1.02 & 0.208 & 0.033 & 0.001 & \phantom{0}1.8\% \\
$\mathcal{A}_3$ & $t=10$ & 0.99 & 0.207 & 0.045 & 0.003 & \phantom{0}2.2\% \\
$\mathcal{A}_2$ & $t=1$ & 2.24 & 1.000 & 0.033 & 0.002 & 93.6\% \\
$\mathcal{A}_2$ & $t=10$ & 1.48 & 1.000 & 0.043 & 0.000 & 57.1\% \\
\hline
\end{tabular}
\end{table}

\section{Detection of \texorpdfstring{$\mathcal{A}_3$}{A3} Against Attack Intensity}\label{app:l3}
Table~\ref{tab:l3sweep} separates the two mechanisms of Corollary~\ref{cor:noadmissible}. The onset rejection rate rises from $0.60$ at $\kappa=1.05$ to unity by $\kappa=1.15$, bracketing the transient threshold of $1.06$ from~\eqref{eq:onsettransient}. The steady-state per-slot rate is at least $0.586$ at every tested intensity, corresponding under the geometric approximation to an expected time to first rejection below two slots, and no run survived the horizon. The runs cannot characterize the tail, while the geometric model extrapolates it under a per-slot independence that the persistent backlog state does not guarantee, so the rate is reported with an interval: at $\kappa=1.05$, $[0.546,0.626]$ implies an expected time to rejection in $[1.60,1.83]$ slots.

\begin{table}[!tb]
\centering
\caption{Detection of $\mathcal{A}_3$ by $L_3$, over twenty calibration seeds with $580$ steady-state comparisons per row and Wilson intervals. Here $\kappa=1$ carries no perturbation and reports the false-positive floor rather than a miss.}
\label{tab:l3sweep}
\footnotesize
\setlength{\tabcolsep}{3pt}
\begin{tabular}{@{}lccccc@{}}
\hline
$\kappa$ & Onset TPR & Steady & 95\% CI & $\mathbb{E}[$slots$]$ & $\Pr[\text{surv.}]^{\dagger}$ \\
\hline
1.00 & 0.000 & 0.000 & $[0.000,0.007]$ & --- & --- \\
1.05 & 0.600 & 0.586 & $[0.546,0.626]$ & 1.71 & $3\times10^{-12}$ \\
1.10 & 0.850 & 0.617 & $[0.577,0.656]$ & 1.62 & $3\times10^{-13}$ \\
1.15 & 1.000 & 0.648 & $[0.609,0.686]$ & 1.54 & $2\times10^{-14}$ \\
1.25 & 1.000 & 0.691 & $[0.653,0.728]$ & 1.45 & $5\times10^{-16}$ \\
1.50 & 1.000 & 0.893 & $[0.865,0.916]$ & 1.12 & $7\times10^{-30}$ \\
3.00 & 1.000 & 1.000 & $[0.993,1.000]$ & 1.00 & $0$ \\
\hline
\multicolumn{6}{@{}p{0.96\columnwidth}@{}}{\footnotesize
$^{\dagger}$Implied by the measured per-slot rate under a geometric model rather than observed, so an extrapolation assuming per-slot independence. No run survived the horizon at any $\kappa>1$.} \\
\hline
\end{tabular}
\end{table}

\section{Latency Scaling, and Sensitivity to the Link-Success Knee}\label{app:scaling}
Table~\ref{tab:scaling} checks the complexity claim of Section~\ref{subsec:composition} by driving each layer in isolation with synthetic observations of the stated size, $400$ repetitions per point after a warm-up pass, with $L_5$ omitted as $O(1)$; absolute figures are not comparable with the in-loop timings of Table~\ref{tab:cost}, and the sweep establishes the growth. Across a fiftyfold increase in $N$ the stack grows by $64.2\times$ against the $115.3\times$ that $N\log N$ predicts, because the sort executes in compiled code while per-layer overhead stays roughly constant. Beyond $N=200$, $L_2/(UN)$ stays near $0.06$~\textmu s, so the cross-modal recomputation of~\eqref{eq:pathloss} rather than the validator's sort governs asymptotic growth, while $L_1$ stays flat because the tag of~\eqref{eq:mac} is computed once per slot.

\begin{table}[!tb]
\centering
\caption{Per-layer latency as the field and the swarm grow, in milliseconds on the platform of Table~\ref{tab:params}. Percentages are of the $5$~s control period.}
\label{tab:scaling}
\footnotesize
\begin{tabular}{@{}rrccccrr@{}}
\hline
$N$ & $U$ & $L_1$ & $L_2$ & $L_3$ & $L_4$ & total & \% $\Delta t$ \\
\hline
\phantom{00}20 & \phantom{0}3 & 0.0053 & 0.0244 & 0.0046 & 0.0129 & 0.0472 & 0.00094 \\
\phantom{00}50 & \phantom{0}5 & 0.0039 & 0.0340 & 0.0054 & 0.0154 & 0.0587 & 0.00117 \\
\phantom{0}100 & 10 & 0.0038 & 0.0663 & 0.0075 & 0.0228 & 0.1003 & 0.00201 \\
\phantom{0}200 & 20 & 0.0039 & 0.2075 & 0.0111 & 0.0374 & 0.2599 & 0.00520 \\
\phantom{0}500 & 30 & 0.0041 & 0.9245 & 0.0221 & 0.0541 & 1.0048 & 0.02010 \\
1000 & 50 & 0.0048 & 2.9023 & 0.0405 & 0.0832 & 3.0308 & 0.06062 \\
\hline
\end{tabular}
\end{table}

Table~\ref{tab:knee} sweeps the link-success knee under a deterministic reference scheduler so that the channel parameter varies in isolation. The degradation ratio rises monotonically with link quality, from $6.4$ to $52.5$, because a healthy scheduler drives the clean cost toward zero while the attacked cost stays high, so the conclusion of Section~\ref{subsec:results} is robust across a factor of eight in $\Gamma$. The same sweep yields $\Gamma=13.8$ at the operating point, close to the published $14.2\times$ and far above the $3.13$ of the deployed agent, which indicates that the larger figure reflects a greedier scheduler rather than a stronger attack.

\begin{table}[!tb]
\centering
\caption{
%Sensitivity of the undefended attacked condition to the link-success knee, under a deterministic reference scheduler so the channel parameter varies in isolation. Absolute values are therefore not comparable with Table~\ref{tab:results}. TTF is the mean number of slots to a sensor's first service, and the last column is the fraction of slots in which the $L_1$ band rejects at $\kappa=3$.
Sensitivity to the link-success knee under the deterministic reference scheduler. TTF is the mean time to first service and $L_1$ rej. the rejection fraction at $\kappa=3$.
}
\label{tab:knee}
\footnotesize
\begin{tabular}{@{}lccccc@{}}
\hline
Link-success knee $\tau$ & $J^{\mathrm{cl}}$ & $J^{\mathcal{A}_1}$ & $\Gamma$ & TTF & $L_1$ rej. \\
\hline
0.36 & 79.5 & 511.7 & \phantom{0}6.4 & 9.1 & 0.18 \\
0.49 & 46.3 & 494.6 & 10.7 & 6.8 & 0.09 \\
0.62$^{\ast}$ & 34.4 & 473.5 & 13.8 & 5.3 & 0.04 \\
0.75 & 22.5 & 452.9 & 20.1 & 4.1 & 0.03 \\
0.91 & \phantom{0}8.2 & 432.6 & 52.5 & 3.3 & 0.02 \\
\hline
\multicolumn{6}{@{}p{0.9\columnwidth}@{}}{\footnotesize
$^{\ast}$Operating point used in the evaluation.} \\
\hline
\end{tabular}
\end{table}

\bibliographystyle{IEEEtran}
\bibliography{references}

\end{document}